\documentclass[authoryear,12pt]{elsarticle}
\usepackage{booktabs}
\usepackage{mathrsfs}
\usepackage{amsmath}
\usepackage{amsthm}
\usepackage{verbatim}
\usepackage{graphicx}
\usepackage{cases}
\usepackage[normal, bf , small]{caption}
\usepackage{bm}
\usepackage{rotating}
\usepackage{bigstrut}
\usepackage{multirow}
\usepackage{natbib}
\usepackage{pdflscape}
\usepackage[page]{appendix}
\makeatletter
\def\ps@pprintTitle{%
  \let\@oddhead\@empty
  \let\@evenhead\@empty
 \def\@oddfoot{\reset@font\hfil\thepage\hfil}
  \let\@evenfoot\@oddfoot
}
\makeatother

\newtheorem{theorem}{Theorem}[section]

\newtheorem{proposition}[theorem]{Proposition}

\newcommand{\ee}{\mathrm{e}}

\newcommand{\esp}{\mathbb{E}}
\newcommand{\espQ}{\mathbb{\tilde{E}}}

\newcommand{\corr}{\mathrm{corr}}

\newcommand{\sgn}{\operatorname{sgn}}
\newcommand{\mgf}{\operatorname{mgf}}
\usepackage{lipsum}
\newcommand{\numAssets}{\mathit{\Upsilon}}
\newcommand{\numWienerProcess}{\mathit{n_w}}
\newcommand{\ProbP}{\mathbb{P}}
\newcommand{\returnP}{\mathrm{\alpha}}
\newcommand{\EJumpP}{\mathrm{\beta}}
\newcommand{\JIntensityP}{\mathrm{\lambda}}
\newcommand{\CompoundPoissonP}{\mathrm{Q}}
\newcommand{\PoissonP}{\mathit{N}}
\newcommand{\WienerP}{\mathit{W}}
\newcommand{\JumpP}{\mathit{Y}}
\newcommand{\PDFJumpP}{\mathit{f^{(i)}(y)}}
\newcommand{\coefficientWienerP}[1]{\mathrm{\gamma_{\mathit{ij}}^{#1}}}
\newcommand{\ProbQ}{\tilde{\mathbb{P}}}
\newcommand{\EJumpQ}{\mathrm{\tilde{\beta}}}
\newcommand{\JIntensityQ}{\mathrm{\tilde{\lambda}}}
\newcommand{\PDFJumpQ}{\mathit{\tilde{f}^{(i)}(y)}}
\newcommand{\logJumpMeanQ}{\mathrm{\eta}}
\newcommand{\logJumpVolatilityQ}{\mathrm{\upsilon}}
\newcommand{\WienerQ}{\mathit{\tilde{W}}}
\newcommand{\shift}{\mathrm{\delta}}
\newcommand{\sign}{\mathit{b}}
\newcommand{\functionF}{\mathit{g}}
\newcommand{\ShiftedBasket}{\mathit{B}}
\newcommand{\ShiftedStrike}{\mathit{K}}
\newcommand{\Basket}{\mathit{B}^*}
\newcommand{\Strike}{\mathit{K^*}}
\newcommand{\lowerLimit}{\mathit{l_1}}
\newcommand{\upperLimit}{\mathit{l_2}}
\newcommand{\parameterA}{\mathit{h_1}}
\newcommand{\parameterB}{\mathit{h_2}}
\newcommand{\zetaTilde}{\tilde{z}}

\newcommand{\coeffHermite}[1]{\varphi_{#1}}
\newcommand{\stocProcess}[1]{\left\{#1\right\}_{t\geq 0}}
\numberwithin{equation}{section}

\newcommand{\COne}{C1} 
\newcommand{\CTwo}{C2} 
\newcommand{\CThree}{C3} 
\newcommand{\CFour}{C4} 
\newcommand{\CFive}{C5} 
\newcommand{\CSix}{C6} 
\newcommand{\CSeven}{C7} 
\newcommand{\CEight}{C8} 
\newcommand{\CNine}{C9} 
\newcommand{\CTen}{C10}
\newcommand{\setOption}{\mathit{O}}
\newcommand\tstrut{\rule{0pt}{3ex}}
\usepackage{amssymb}

\begin{document}

\begin{frontmatter}
\title{Pricing and Hedging Basket Options with Exact Moment Matching}

\author[kent]{Tommaso Paletta\corref{cor1}}
 \ead[]{t.paletta@kent.ac.uk}
\author[unical]{Arturo Leccadito}
\author[kent]{Radu Tunaru}
\address[kent]{Business School, University of Kent, Park Wood Road, Canterbury CT2 7PE, UK,}
\address[unical]{Dipartimento di Economia, Statistica e Finanza, Universit\`a della Calabria, Ponte Bucci cubo 3C, Rende (CS), 87030, Italy }
\cortext[cor1]{Corresponding author:}

\begin{abstract}
Theoretical models applied to option pricing should take into account the empirical characteristics of the underlying financial time series.
In this paper, we show how to price basket options when assets follow a shifted log-normal process with jumps capable of accommodating negative skewness. Our technique is based on the Hermite polynomial expansion that can match exactly the first $m$ moments of the model implied-probability distribution. This method is shown to provide superior results for basket options not only with respect to pricing but also for hedging.
\end{abstract}

\begin{keyword}
Basket options \sep Shifted log-normal jump process \sep Hermite polynomials \sep Negative skewness \sep Option pricing and hedging \\\vspace{0.2cm}
\\\vspace{0.3cm}
Submitted online: \today
\end{keyword}

\end{frontmatter}\newpage

\section{Introduction}
Basket options are contingent claims on a group of assets such as equities, commodities, currencies and even other vanilla derivatives. Spread options can be conceptualised as basket options whose payoffs depend on the price differential of two assets.
Basket options are a subclass of exotic options commonly traded over-the-counter in order to hedge away exposure to correlation or contagion risk. Hedge-funds also use them for investment purposes, to combine diversification with leveraging. Spread options are heavily traded on the commodity markets, in particular on energy markets, where several final products are industrially produced from the same raw material.

From a modelling point of view, the  framework  ought to be  multidimensional since baskets of 15 to 30 assets are frequently traded. Many pricing models that seem to work well for single assets cannot be easily expanded to a multidimensional set-up, mainly due to computational difficulties. Hence, in order to circumvent these difficulties, practitioners resort to classic multidimensional geometric Brownian motion type models which can be easily implemented. However, by doing so, the empirical characteristics of the assets in the basket are simply ignored. In particular, negative skewness, which is well known to characterize equities, cannot be captured properly by these simple models which can produce a limited  range of values for skewness. Recently, \cite{BorovkovaPermana2007} and \cite{BorovkovaPermanaWeide2007,BorovkovaPermanaWeide2012}
have proposed a new methodology that can incorporate negative skewness while still retaining analytical tractability, under a shifted log-normal distribution, by considering the entire basket as one single asset.  This strong assumption allowes the derivation of closed-form formulae for option pricing.

Ideally, one would like the best of both worlds, realistic modelling and precise calculations.
In this paper, we present a general computational solution to the problem of multidimensional models which lack closed-form formulae or models that require burdensome numerical procedures.  The shifted log-normal process with jumps exemplifies the problem encountered with pricing basket options. On one hand, this distribution is very useful to follow the dynamics of one asset, but on the other hand expanding this set-up to a basket of assets leads to severe computational problems. We circumvent this problem by employing the Hermite polynomial expansion which is matching exactly the first $m$ moments of the model implied probability distribution. Hence, the only prerequisite of our method is to be able to calculate the moments of the basket in closed form. In addition, the same technique can be applied for any other similar modelling situations for other models. Furthermore,  our methodology is applicable to the situation when some assets in the basket follow one diffusion model and other assets follow a different diffusion model.

The article is structured as follows. In Section~\ref{sec:litrev}, we briefly review the methods proposed for pricing basket and spread options, focusing on approximation techniques. Section~\ref{sec:modfra} contains a description of the continuous-time models we employed here. Our new methodology is discussed in Section~\ref{sec:pricinghedging} and the empirical results are presented in Section~\ref{sec:empres}. The final section concludes.

\section{Related Literature}
\label{sec:litrev}
The number of papers covering basket options and, in particular, spread options has increased considerably in the last three decades.
\cite{Margrabe1978} was the first to develop an exact formula for European spread options when the two assets are assumed to follow a geometric Brownian motion.
\cite{CarmonaDurrleman2003} presented an extensive literature review on pricing methods for spread options as well as introducing a new method.
The methods used to price basket options can be classified into analytical, purely numerical and a hybrid semi-analytical  class based on various expansions and moment matching techniques. Our method belongs to the last category.

By analogy to early papers on pricing Asian options, \cite{Gentle1993} proposed pricing basket options\footnote{In that paper it is assumed that all assets in the basket have positive weights.} by approximating the arithmetic weighted average with its geometrical-average counterpart so that a Black-Scholes type formula could be applied. \cite{KornZeytun2013} improved this approximation using the fact that, if the spot prices of assets in the basket are shifted by a large scalar constant $C$, their arithmetic and geometric means converge asymptotically. They consider log-normally distributed assets and approximate the $C$-shifted distribution by standard log-normal distributions.
\cite{Kirk1995} developed a technique for pricing a spread option by coupling the asset with negative weight with  the strike price, considering their combination as one asset having a shifted distribution and then employing the \cite{Margrabe1978} formula for exchanging two assets. This shift assumption corresponds to a linear approximation of the exercise boundary\footnote{The exercise boundary is the minimal standardized log-price of the first asset that makes the spread option in-the-money as a function of the standardized log-price of the second asset.}. The method in \cite{DengEtAl2008} can be considered as an extension of \cite{Kirk1995}. They derived a closed-form pricing formula for spread options  by applying a quadratic Taylor expansion of the exercise boundary. These results were further extended by \cite{LiEtAl2010} to the  case of $N$ assets with positive and negative  weights. 

\cite{VentramananAlexander2011} and \cite{AlexanderVentramanan2012} priced spread options  and more general multi-asset options (basket and rainbow options) using a portfolio of compound exchange options (CEO). Their idea was to utilise exact replicating portfolios and then approximate the formulae to price the CEOs. Remarkably, \cite{VentramananAlexander2011} derived an analytical formula for American spread options using the early exercise premium approach proposed in \cite{Kim1990}. \cite{BjerksundStensland2011} also priced spread options by direct use of the implied exercise boundary in \cite{Kirk1995}.

When analytical formulae are difficult to find under a particular model, it is common, in the finance industry, to resort to  Monte Carlo (MC) methods. Control variate techniques for pricing basket options are described in \cite{Pellizzari1998} and \cite{KornZeytun2013}. \cite{Barraquand1995} advanced a very general framework to price multidimensional contingent claims by Monte Carlo simulation and quadratic re-sampling. Monte Carlo simulation was also successfully used  to price American style basket options by \cite{Barraquand1995StatePartitioning}, \cite{LongstaffSchwartz2001} and \cite{broadie2004stochastic}.

While Monte Carlo methods offer a feasible solution, the computational cost may be too high even for standard-size baskets commonly traded on the financial markets. Hence, the bulk of the literature on basket option pricing gravitates around approximation methods that circumvent the numerical problems generated by the high-dimensionality of basket models. A typical example is the research by \cite{Li2000} who employed an Edgeworth expansion of a four-parameter skewed generalized-t distribution. Edgeworth series expansions were proposed first by \cite{JarrowRudd1982} and \cite{TurnbullWakeman1991}  to price European basket options and arithmetic Asian options  respectively. \cite{Rubinstein1998} combined an Edgeworth expansion and a binomial tree to price American-style option with pre-specified skewness and kurtosis. This method has two disadvantages. Firstly, the matching of skewness and kurtosis is not exact given that a rescaling of probability  is necessary. Secondly, not all combinations of skewness and kurtosis can be matched because negative probabilities and multi-modal distribution may result.

\cite{Levy1992} approximated the distribution of a basket by matching its first two moments with the moments of a log-normal density function, and consequently  a Black-Scholes pricing formula could be employed.
Other  works improved the log-normal approximation allowing for improved skewness and kurtosis calibration. The displaced diffusion introduced by \cite{Rubinstein1983} considers the shifted basket value as being log-normally distributed. \cite{BorovkovaPermanaWeide2007}, henceforth \emph{BPW}, proposed a generalized log-normal approach that is superior to the model in  \cite{Rubinstein1983} because it allows distributions of a basket to cover  negative values and  negative skewness. \cite{ZhouWang2008} advocated a method similar to that of $BPW$, selecting the log-extended-skew-normal as the approximating distribution. They obtained a Black-Scholes type pricing formula where the standard extended-skew-normal cumulative distribution function replaces the normal one. \cite{BorovkovaPermanaWeide2012} extended this method to price American-style basket options via  a one-dimensional binomial tree.
In an interesting application, \cite{BorovkovaPermana2007} adapted  the method described in  $BPW$ to price Asian basket options.

\cite{MilevskyPosner1998} used the reciprocal gamma distribution to approximate a positively weighted sum of correlated log-normal random variables. Matching the first two moments of the basket, they priced European basket options by a Black-Scholes type formula where the normal cumulative function is substituted by the cumulative distribution function of the gamma distribution. This method returns good results only when the basket has a decaying correlation structure (similar to the one for Asian option). \cite{MilevskyPosner1998b} derived a closed-form pricing formula by using two distributions from the Johnson system of distributions that match the first four moments of the basket value. Asian and basket options prices were calculated by \cite{Ju2002} using the Taylor's expansion for the ratio of the characteristic function of the value of the basket at maturity to that of the approximating log-normal random variable. 
While the literature on pricing basket options is large there is sparse research on calculating the hedging parameters for basket options. \cite{HurdZhou2010} price spread options for two or more assets and also derive the Greek parameters by using fast Fourier transform. The only assumption for the underlying asset price processes is that the characteristic function of the joint return is known analytically.

\section{The Modeling Framework \label{sec:modfra}}
In this paper, we consider a new process for asset prices: the shifted jump-diffusion process. We firstly describe the standard jump-diffusion model in Section \ref{JumpDiffusion} that provides the platform for designing the shifted jump-diffusion model in Section \ref{shiftedModel}.

\subsection{Jump-diffusion Model}
\label{JumpDiffusion}
Consider the filtered probability space\footnote{The contents and notation in this subsection benefit from  \cite[chap. 11.5]{shreve2004stochastic}.} $(\Omega,\mathscr{F},(\mathscr{F}_t)_{0\leq t \leq T},\ProbP)$. Let us define, on this space, the financial market consisting of $\numAssets$ assets, $S^{(i)}$ for any $i=1,\cdots,\numAssets$, with dynamics given by
\begin{equation}
\label{assedDinamics}
d S_t^{(i)}=(\returnP_i-\EJumpP_i \JIntensityP_i)S_t^{(i)} dt + S_t^{(i)} \sum_{j=1}^{\numWienerProcess}\coefficientWienerP{} d \WienerP_t^{(j)}+S_{t^-}^{(i)} d\CompoundPoissonP_t^{(i)}, \quad i=1,\cdots, \numAssets
\end{equation}
and the bank account
\begin{equation}
\label{bankAccount}
dM_t = rM_t dt
\end{equation}
that can be used to borrow and deposit money with continuously compounded interest rate $r\geq 0$, assumed constant over time.

Equation~\eqref{assedDinamics} describes a jump-diffusion process where $\returnP_i$ is the expected rate of return on the asset $i$, $\stocProcess{\WienerP_t^{(j)}}$ are $\numWienerProcess$ mutually independent Wiener processes,  $\stocProcess{\CompoundPoissonP_t^{(i)}}$  are independent compound Poisson processes formed from some underlying Poisson processes $\stocProcess{\PoissonP_t^{(i)}}$ with intensity $\JIntensityP_i\geq 0$ and $\JumpP_j^{(i)}$ representing the jump amplitude  of the $j$-th jump  of $\PoissonP_t^{(i)}$ for any $i=1,\cdots,\numAssets$.  The jumps $\JumpP_j^{(i)}$ for any $i=1,\cdots,\numAssets$ are independent and identically distributed random variables with probability density function $\PDFJumpP:\Re^{+}\rightarrow[0,1]$ and expected value under the physical measure $\EJumpP_i=\esp[\JumpP^{(i)}]=\int_{\Re}y \PDFJumpP dy$. 
Moreover, jumps for different assets are independent.

Applying standard Ito's rule for jump processes \cite[see][Chap. 11.7.2]{shreve2004stochastic}, it is possible to derive a closed-form solution for the SDEs in \eqref{assedDinamics} as:
\begin{equation}
\label{assetSolution}
S_t^{(i)}=S_0^{(i)}e^{\left(\returnP_i-\EJumpP_i \JIntensityP_i-\frac{1}{2}\sum_{j=1}^{\numWienerProcess}\coefficientWienerP{2} \right)t+\sum_{j=1}^{\numWienerProcess}\coefficientWienerP{} \WienerP_t^{(j)}}  \prod_{l=1}^{\PoissonP_l^{(i)}}{(\JumpP_j^{(i)}+1)}, \quad i=1,\cdots, \numAssets.
\end{equation}

The market given by \eqref{assedDinamics} and \eqref{bankAccount} is arbitrage free if and only if there exists $\bm{\theta}=[\theta_1,\cdots,\theta_\numWienerProcess]$, $\tilde{\bm{\beta}}=[\tilde{\beta}_1,\cdots, \tilde{\beta}_\numAssets]$ and $\tilde{\bm{\lambda}}=[\tilde{\lambda}_1,\cdots, \tilde{\lambda}_\numAssets]$ solving the system of market price of risk equations
\begin{equation}
\label{systemArbitrageFree}
\returnP_i-\EJumpP_i \JIntensityP_i -r=\sum_{j=1}^{\numWienerProcess}\coefficientWienerP{} \theta_j- \EJumpQ_i \JIntensityQ_i, \quad i=1,\cdots, \numAssets.
\end{equation}
The solution to \eqref{systemArbitrageFree} is, in general, not unique. Nevertheless, we assume that one solution of the system~\eqref{systemArbitrageFree} is selected\footnote{There is a large literature devoted to the issue of selecting a pricing measure. For a review, see \cite{Frittelli2000} and references within.} and a pricing measure $\ProbQ$ is fixed\footnote{Henceforth, $\esp$ and $\espQ$ are used to indicate the expectation operators under the physical measure $\ProbP$ and under the risk-neutral measure $\ProbQ$, respectively.}. Under the $\ProbQ$-measure, for asset $i$-th in the basket, we still have the compound Poisson processes $\stocProcess{\CompoundPoissonP_t^{(i)}}$, the underlying Poisson process $\stocProcess{\PoissonP_t^{(i)}}$ and the jumps $\JumpP_j^{(i)}$ but now the intensity of the Poisson process $\stocProcess{\PoissonP_t^{(i)}}$ is $\JIntensityQ_i$ and $\EJumpQ_i=\espQ[\JumpP^{(i)}]=\int_{\Re}y \PDFJumpQ dy$. One way to model the size of the jumps is taking, for each asset, jumps iid log-normally distributed\footnote{When we impose a log-normal distribution for $\JumpP^{(i)}_j+1$, we implicitly assume that the system of equations in \eqref{systemArbitrageFree} has a solution. Furthermore, any other distribution $\PDFJumpQ:\Re^{+}\rightarrow[0,1]$ could have been chosen, if it  leads to a feasible system.} such that $\espQ[\log(\JumpP_j^{(i)}+1)]=\logJumpMeanQ_i$ and $\widetilde{Var}[\log(\JumpP_j^{(i)}+1)]=\logJumpVolatilityQ_i^2$.

The risk-neutral $\ProbQ$-dynamics of the assets composing the basket can be described as:
\begin{equation}
\label{assedDinamicsQ}
d S_t^{(i)}=(r-\EJumpQ_i \JIntensityQ_i)S_t^{(i)} dt + S_t^{(i)} \sum_{j=1}^{\numWienerProcess}\coefficientWienerP{} d \WienerQ_t^{(j)}+S_{t^-}^{(i)} d\CompoundPoissonP_t^{(i)}, \quad i=1,\cdots, \numAssets
\end{equation}
where $\stocProcess{\WienerQ_t^{(i)}}$ are independent  Wiener processes under the martingale measure $\ProbQ$.

The solutions to \eqref{assedDinamicsQ} can be derived in the following convenient closed-form:
\begin{equation}
\label{solutionArbitrageFree}
S_t^{(i)}=S_0^{(i)}e^{\left(r-\EJumpQ_i \JIntensityQ_i-\frac{1}{2}\sum_{j=1}^{\numWienerProcess}\coefficientWienerP{2}\right)t+\sum_{j=1}^{\numWienerProcess}\coefficientWienerP{} \WienerQ_t^{(j)}}  \prod_{l=1}^{\PoissonP_t^{(i)}}{(\JumpP_l^{(i)}+1)},\quad i=1,\cdots,\numAssets.
\end{equation}

\subsection{Shifted jump-diffusion Model}
\label{shiftedModel}
From a modelling point of view, it would be more appropriate to use models that are capable of generating negative skewness reflecting the empirical evidence in equity markets. One such flexible model is the generalized GBM process in  \cite{BorovkovaPermanaWeide2007}. Here, we extend that model to include jumps, thus obtaining a jump-diffusion process for the displaced or shifted asset value:
\small
\begin{eqnarray}
\label{assedDinamicsShifted}
d \left(\sign_i S_t^{(i)}-\shift_t^{(i)}\right)&=&(\returnP_i-\EJumpP_i \JIntensityP_i)\left(\sign_i S_t^{(i)} -\shift_t^{(i)}\right)dt + \left(\sign_i S_t^{(i)}-\shift_t^{(i)}\right) \sum_{j=1}^{\numWienerProcess}\coefficientWienerP{} d \WienerP_t^{(j)}+\nonumber\\
&+&\left(\sign_i  S_{t^-}^{(i)}-\shift_t^{(i)}\right) d\CompoundPoissonP_t^{(i)},\quad i=1,\cdots, \numAssets.
\end{eqnarray}
In~\eqref{assedDinamicsShifted},  $\shift_t^{(i)}$ is the shift applied to $S_t^{(i)}$ at time $t$ and $\sign_i \in\{-1,1\}$. We assume that  $\sign_i$ is negative when the asset price assumes values in $(-\infty,-\shift_t^{(i)})$ and positive when the range for the asset price is $(\shift_t^{(i)},\infty)$. The shift $\shift_t^{(i)}$ is assumed to follow equation $d\shift_t^{(i)}=r\shift_t^{(i)}dt$, with $\shift_0^{(i)}\in \Re$ and, consequently, represents the cash position at time 0.  All the other parameters  have the same meaning as described above for equation~\eqref{assedDinamics} only that they refer now to the shifted asset prices.

The solution of equation~\eqref{assedDinamicsShifted}, under the risk-neutral pricing measure $\ProbQ$, is clarified in the following proposition\footnote{A more general version of Proposition \ref{propositionShift} is stated in the Proposition~\ref{generalprop}  for the sake of completeness, but it is not used empirically in this paper.}.

\begin{proposition}
\label{propositionShift}
Consider that the assets in a basket follow the shifted jump-diffusion model with dynamics given by the SDE \eqref{assedDinamicsShifted}
with the shifting process $\stocProcess{\shift_t^{(i)}}$
satisfying $d\shift_t^{(i)} = r\shift_t^{(i)} dt$. If a solution $(\bm{\theta},\tilde{\bm{\beta}},\tilde{\bm{\lambda}})$
of the system \begin{equation}
\label{systemArbitrageFreeProposition}
\returnP_i-\EJumpP_i \JIntensityP_i -r=\sum_{j=1}^{\numWienerProcess}\coefficientWienerP{} \theta_j- \EJumpQ_i \JIntensityQ_i, \quad i=1,\cdots, \numAssets
\end{equation}
does exist and is selected in association with the risk-neutral pricing measure
$\ProbQ$, then, under this risk-neutral measure,
\begin{equation}
\label{assetSolutionShiftedFirst}
S_t^{(i)}=\left(S_0^{(i)}-\sign_i \shift_0^{(i)}\right)e^{\left(r-\EJumpQ_i \JIntensityQ_i-\frac{1}{2}\sum_{j=1}^{\numWienerProcess}\coefficientWienerP{2}\right)t+\sum_{j=1}^{\numWienerProcess}\coefficientWienerP{} \WienerQ_t^{(j)}}  \prod_{l=1}^{\PoissonP_t^{(i)}}{(\JumpP_l^{(i)}+1)}+\sign_i\shift_0^{(i)} e^{r t}.
\end{equation}

\end{proposition}
\begin{proof}
See Appendix~\ref{AppendixShif}
\end{proof}

In order to simplify the notation for the empirical work carried out in Section \ref{comparison}, we denote $V_t^{(i)}=\sum_{j=1}^{\numWienerProcess}{\frac{\coefficientWienerP{}}{\sigma_i} \WienerQ_t^{(j)}}$ where $\sigma_i^2=\sum_{j=1}^{\numWienerProcess}{\coefficientWienerP{2}}$. Thus  $\stocProcess{V_t^{(i)}}$ are dependent standard Brownian motions with \[\rho_{l_1l_2}=\corr(V_t^{(l_1)},V_t^{(l_2)})=\frac{1}{\sigma_{l_1}\sigma_{l_2}}\sum_{j=1}^{\numWienerProcess}{\gamma_{\mathit{l_1j}}\gamma_{\mathit{l_2j}}},\] and consequently
\begin{equation}
\label{assetSolutionShifted}
S_t^{(i)}=\left(S_0^{(i)}-\sign_i \shift_0^{(i)}\right)e^{\left(r-\EJumpQ_i \JIntensityQ_i-\frac{1}{2}\sigma_i^2\right)t+\sigma_i V_t^{(i)}}  \prod_{l=1}^{\PoissonP_t^{(i)}}{(\JumpP_l^{(i)}+1)}+\sign_i\shift_0^{(i)} e^{r t}
\end{equation}
    is used instead of \eqref{assetSolutionShiftedFirst}.

Finally, we point out that the shifted jump-diffusion may encompass three sub-cases:
\begin{itemize}
\item geometric Brownian motion (GBM) when $\shift_0^{(i)}=0$ and $\JIntensityQ_i=0$ for each asset $i$;
\item shifted GBM when $\JIntensityQ_i=0$ for each asset $i$;
\item standard jump-diffusion when $\shift_0^{(i)}=0$ for each asset $i$.
\end{itemize}

\section{Pricing and hedging methodology\label{sec:pricinghedging}}

Our aim is to price European basket options under the shifted jump-diffusion model. The payoff at maturity
of such option is $(\Basket_T-\Strike)^+$, driven by the underlying variable
\begin{equation}\label{basket}
\Basket_t =\sum_{i=1}^{\numAssets} a_i S_t^{(i)},
\end{equation}
where $\Strike$ is the strike price, $\bm{a}  =(a_1,\ldots,a_\numAssets)'$ is the vector of basket weights, which could be positive or negative, and $T$ is the time to maturity.

Under the majority of models applied in practice,  the probability density of the basket $\Basket_t$ cannot usually be obtained in closed-form. The methodology proposed here is circumventing this problem using a Hermite approximation probability density that will replace the risk-neutral density implied by the model~\eqref{assetSolutionShifted}. In addition, the approximation density derived in this paper is constructed in such a way to match exactly up to the first $m$ moments of the model implied risk-neutral density.

\cite{LeccaditoToscanoTunaru2012} proposed the Hermite tree method for pricing financial derivatives. In a nutshell, the idea  is to match the moments of the log-returns of the underlying asset with  the moments of a discrete random variable. This work elaborates on some variants of the method presented in \cite{LeccaditoToscanoTunaru2012} to deal with baskets that may take on negative values. In particular, the binomial distribution has been changed with the asymptotically equivalent Gaussian distribution (coded as \emph{G}) and the moment matching is done on two different types of return quantities (coded as \emph{A} and \emph{B}) as specified in Table \ref{momentMatching}, where $\ShiftedBasket_T$ is defined by equation \eqref{shiftedBasket}. Henceforth, $\ShiftedBasket_0$ is assumed to be different from  0.

\begin{center}
[Table \ref{momentMatching} about here.]
\end{center}

\subsection{Moments of the baskets}
The first step in our methodology is to derive the moments of the basket~\eqref{basket} under the specification of a  model for the underlying assets. For model \eqref{assetSolutionShifted},  consider the ``shifted basket''
\begin{eqnarray}
\label{shiftedBasket}
\ShiftedBasket_t&=&\sum_{i=1}^{\numAssets}{a_i \left(S_t^{(i)}- \sign_i \shift_0^{(i)}\ee^{rt} \right)}
\end{eqnarray}
and the ``shifted strike price''
\begin{eqnarray}
\label{shiftedStrikePrice}
\ShiftedStrike&=& \Strike-\sum_{i=1}^{\numAssets}{a_i \sign_i \shift_0^{(i)}\ee^{rt}}.
\end{eqnarray}
For practical purposes we shall calculate the moments of the shifted basket.  Proposition~\ref{prop:moments} shows how to calculate these moments.
\begin{proposition}\label{prop:moments}
The $k$-moment of $\ShiftedBasket_t$, under $\ProbQ$, is given by  
\begin{eqnarray}
\label{moments}
\mu_k &=& \espQ[\ShiftedBasket_t^k]
=\sum_{i_1=1}^\numAssets\cdots\sum_{i_k=1}^\numAssets
 a_{ i_1 } \left(S_0^{(i_1)}-\sign_{i_1}\shift_0^{(i_1)}\right)  \ee^{  (r+\omega_{i_1})t    }\times \cdots\nonumber\\
&\cdots&\times  a_{ i_k } \left(S_0^{(i_k)}-\sign_{i_1}\shift_0^{(i_k)}\right)  \ee^{  (r+\omega_{i_k})t   }
\mgf(\bm{e}_{ i_1 } + \ldots + \bm{e}_{ i_k })
\end{eqnarray}
\normalsize
where $\omega_j=-\EJumpQ_j \JIntensityQ_j-\frac{1}{2}\sigma_j^2$, $\bm{e}_j$ is the vector having 1 in position $j$ and zero elsewhere. Furthermore, the moment generation function of $\sigma_i V_t^{(i)}  +\sum_{l=1}^{\PoissonP _t^{(i)}}{\log{(\JumpP_l^{(i)}+1)}}$ is given by
\begin{equation}\label{MultivariateMerton spec mgf}
\mgf(\bm{u})=
\exp\left\{ t\bm{u}'\bm{\Sigma}\bm{u}/2\right\} \prod_{i=1}^\numAssets \mgf_{N_t^{(i)}}\left(\logJumpMeanQ_{i} u_i + \logJumpVolatilityQ^2_{i}u_i^2/2  \right)
\end{equation}
where $\bm{\Sigma}$ denotes the covariance matrix of $\bm{V}=\left(V_t^{(1)},\cdots,V_t^{(\numAssets)}\right)'$,
and
 \begin{equation}
 \mgf_{\PoissonP_t^{(i)}}(u)=\exp(t\JIntensityQ_{i} (\ee^u-1)).
 \end{equation}
\end{proposition}
\begin{proof}
See Appendix \ref{proof41}.
\end{proof}

\subsection{European Basket Call option pricing and hedging}

The mechanism of shifting the basket and strike price in equations \eqref{shiftedBasket} and \eqref{shiftedStrikePrice} allows rewriting the European basket call option price in two equivalent ways:
\begin{equation}
\label{pricingFormula}
c=\ee^{-rT}\espQ[(\Basket_T-\Strike)^+]=\ee^{-rT}\espQ[(\ShiftedBasket_T-\ShiftedStrike)^+].
\end{equation}
We are going to use two Hermite approximation variants\footnote{The methodologies described in this paper are supported by various computational tools that are described in  Appendix \ref{AppendixA} for internal consistency.} described in Table \ref{momentMatching}, each variant being associated with a particular target quantity for the basket. 

The next proposition provides a formula for  the European call basket option price under the approximations considered in this paper.

\begin{proposition}\label{propositionPrice}
The price of a European call basket option with the Hermite expansion variant $mGA$ or $mGB$ is given by:
\begin{equation}\label{solutioncompact}
c_0=\ShiftedBasket_0\left[(\coeffHermite{0}+\parameterA)\Phi(-\parameterB \zetaTilde)+\parameterB \functionF(\zetaTilde)\right]-\ShiftedStrike \ee^{-rT}\Phi(-\parameterB \zetaTilde)
\end{equation}
where
\begin{equation}\label{definitionf}
\functionF(\zetaTilde)=\phi(\zetaTilde)\sum_{k=0}^{m-2}{\coeffHermite{k+1}H_k(\zetaTilde)},
\end{equation}
 $K$ is the shifted strike price, $\parameterA=0$ for the variant $mGA$ and $\parameterA=1$ for the variants $mGB$,  $\parameterB=\sgn(\ShiftedBasket_0)$, $\zetaTilde$ is the solution of $[J(\zetaTilde)+\parameterA]\ShiftedBasket\ee^{rT}=\ShiftedStrike$, $\phi(\cdot)$ is the standard normal density function and $\Phi(\cdot)$ is the standard normal cumulative distribution function.
\end{proposition}
\begin{proof}
See Appendix \ref{proof42}
\end{proof}

The next proposition reports the formula for the hedging parameter with respect to the variable $u$, which can be any of the quantities $S^{(i)}_0$, $\ShiftedBasket_0$, $\sigma_i$, $r$, $T$, $a_i$, $\JIntensityQ_i$, $\shift_0^{(i)}$, $\EJumpQ_i$, $\logJumpMeanQ_i$ or $\logJumpVolatilityQ_i$.

\begin{proposition}
\label{proposition43}
For $c_0$, $\parameterA$, $\parameterB$, $\zetaTilde$, $\functionF(\cdot)$, $\phi(\cdot)$ and $\Phi(\cdot)$ defined in Proposition \ref{propositionPrice}, the hedging parameter of a European call basket option, with respect to the variable $u$,  under the Hermite expansion variant $mGA$ or $mGB$, is given by
\begin{eqnarray}\label{genericGreek}
\frac{\partial c_0}{\partial u}&=& c_0 \ee^{rT}\frac{\partial \ee^{-rT}}{\partial u}+\ShiftedBasket_0\left[\parameterB \functionF'(\zetaTilde)+\frac{\partial \coeffHermite{0}}{\partial u}\phi(-\parameterB \zetaTilde)\right]+ \nonumber\\&+&\ee^{-rT} \frac{\partial (\ShiftedBasket_0\ee^{rT})}{\partial u}\left[\parameterB \functionF(\zetaTilde)+\coeffHermite{0}\phi(-\parameterB\zetaTilde)+\parameterA\left(-\parameterB\Phi(\zetaTilde)+\frac{\parameterB+1}{2}\right)\right] \nonumber \\
\end{eqnarray}
where
\begin{equation}
\functionF'(\zetaTilde)=\phi(\zetaTilde)\sum_{k=0}^{m-2}{\frac{\partial \coeffHermite{k+1}}{\partial u}H_k(\zetaTilde)},
\end{equation}
\end{proposition}
\begin{proof}
See Appendix \ref{proof43}
\end{proof}

In Section \ref{HedgingComparison}, a comparison of our method with other methods in the literature is carried out using the Delta-hedging performances as a yardstick. For that exercise, it is particularly important to apply formula~\eqref{genericGreek} for the case when $u=\ShiftedBasket_0$:
\begin{eqnarray}\label{DeltaGreek}
\frac{\partial c_0}{\partial u}&=& \ShiftedBasket_0\left[\parameterB \functionF'(\zetaTilde)+\frac{\partial \coeffHermite{0}}{\partial u}\phi(-\parameterB \zetaTilde)\right]+ \nonumber\\&+&\parameterB \functionF(\zetaTilde)+\coeffHermite{0}\phi(-\parameterB\zetaTilde)+\parameterA\left(-\parameterB\Phi(\zetaTilde)+\frac{\parameterB+1}{2}\right).
\end{eqnarray}

\section{Empirical Comparisons\label{sec:empres}}
\label{comparison}
\subsection{Pricing performances}
The usefulness of a newly proposed method can be gauged by comparing it with other established methods in the literature. To this end, in this section, the two methods $mGA$ and $mGB$ of our Hermite approximation approach are directly benchmarked with the method in \cite{BorovkovaPermanaWeide2007}. In addition, the Monte Carlo with control variate methodology outlined in \cite{Pellizzari1998} is adapted to deal with assets having the dynamic specified by equation \eqref{assetSolutionShifted}. The model performance is determined considering three measures of error:
\begin{description}
\item[\COne] \textit{number of best solutions found}, defined as number of times the minimum squared error is reached under the specified method\footnote{Throughout this paper $1\{\cdot\}$ will denote the indicator function given, for any set $A$, by $1\{A\}(x)=\left\{
                                                                                               \begin{array}{ll}
                                                                                                 1, & \hbox{if $x \in A$;} \\
                                                                                                 0, & \hbox{otherwise.}
                                                                                               \end{array}
                                                                                             \right.
    $.}:
\begin{equation}
\mbox{\COne}_l=\sum_{i\in \setOption}{1\left\{\min_{j\in \{\mathit{BPW,mGA,mGB}\}}\mathit{SE}_i^j=\mathit{SE}_i^l\right\}}
\end{equation}
where $\setOption$ is the set of options considered and, for each option $i\in \setOption$, $\mathit{SE}_i^j$ is the squared error for option $i$ and method $j \in \{BPW, mGA, mGB\}$;
\item[\CTwo] \textit{number of times a method is not able to price an option}, that is the procedure of  moment matching gives poor results for the option. We consider the moment matching to be poor when the relative error ($\mathit{E_r}$) is greater than 5\%:
\begin{equation}
\mbox{\CTwo}_l=\sum_{i\in \setOption}{1\{\mathit{E_r}_i^l>5\%\}}
\end{equation}
By convention, if \CTwo$\,$  is not explicitly stated, it is equal to 0.
\item[\CThree] \textit{square root of $\mathit{MSE}$}, calculated only relative to the options for which the method was able to find a numerical solution.
\end{description} 
	\subsubsection{Multi-dimensional Model Comparisons}
This section is a direct comparison with the method in \cite{BorovkovaPermanaWeide2007}. The six basket options priced in that paper are summarized in Table  \ref{tab:baskets}.  The special case $\JIntensityQ_i=0$, $\shift^{(i)}_0=0$ and $\sign_i=1$ combined with equation~\eqref{assetSolutionShifted} falls onto the GBM case for all assets in the basket.
Table \ref{tab:GBMresults} contains the comparison results. The prices obtained here for the shifted log-normal model of $BPW$  are different from the ones in \cite{BorovkovaPermanaWeide2007} because, to be consistent with the other models in the paper, we are pricing basket options where the underlying assets are the stock and not the forward contracts.

The empirical results indicate that the method $6GA$ appears to be the best method according to \COne. The methods $4GA$ and $4GB$ give, for these six basket options, exactly the same prices and under the  \CThree$\,$ (RMSE),  they achieve the best performance. For the baskets analysed here, there is very little advantage in matching all six moments, the Hermite approximation method working as well when only the first four moments are matched.

\begin{center}
[Table \ref{tab:baskets} about here.]
\end{center}

\begin{center}
[Table \ref{tab:GBMresults} about here.]
\end{center} 

	\subsubsection{Comparison under a set of simulated scenarios}
\label{robustComparison}

A general comparison is performed considering a set of 2000 randomly generated options. In particular, the parameters of the underlying model \eqref{assetSolutionShifted} are drawn as follows:
\begin{itemize}
\item the risk-free rate $r$ is uniformly distributed between 0.0 and 0.1;
\item the volatility parameters $\sigma_i$ are uniformly distributed between 0.1 and 0.6;
\item the time-to-maturity $T$ is uniformly distributed between 0.1 and 1 years; 
\item current spot prices $S_0^{(i)}$ are uniformly distributed between 70 and 130;
\item the weights $a_i$ of the assets  in the basket are uniformly distributed between -1 and 1;
\item the ratios $\Strike$ over $\Basket_0$ are uniformly distributed between 0.95 and 1.05;
\item the shifts $\shift_0^{(i)} \ee^{rT}$ range uniformly between -20 and 20;
\item each asset has the same probability to be positively  ($\sign_i=1$) or negatively ($\sign_i=-1$) shifted;
\item the intensities of the Poisson processes $\JIntensityQ_i$ are uniformly distributed between 0 and 0.2;
\end{itemize}
For each scenario, the correlation matrix is randomly generated satisfying the semi-positiveness condition.
Furthermore, the option prices scenarios are divided into two sets of 1000 options each:
\begin{description}
\item[Set 1]  includes 500 options with the number of assets uniformly distributed between 2 and 10, 300 options with the number of assets uniformly distributed between 11 and 15, 100 options between 16 and 20 and 100 options between 21 and 50. Each asset has jumps with average size ($\logJumpMeanQ$) uniformly distributed between -0.3 and 0, and volatility ($\logJumpVolatilityQ$) uniformly distributed between 0 and 0.3;
\item[Set 2] includes 1000 options with the number of assets uniformly distributed between 2 and 50, each asset having jumps with average size ($\logJumpMeanQ$) uniformly distributed between -0.3 and 0.3, and volatility ($\logJumpVolatilityQ$) uniformly distributed between 0 and 0.3.
\end{description}
For baskets with less than 10 assets in Set 1, results are calculated matching  $m=4$ moments and also matching $m=6$ moments. As shown in Table \ref{tab:randomOptions2-10}, the results for $m=6$ are outperformed by the results with $m=4$. This is consistent with the research  of \cite{CorradoSu1997}  who concluded that considering more than four moments ``creates severe collinearity problems since all even $\ldots$ (moments) $\ldots$ are highly correlated with each other $\ldots$ (and) $\ldots$ similarly, all odd-numbered subscripted (moments) are also highly correlated''. Therefore, for baskets with more than 10 assets in Set 1 and also for all the options in Set 2, we conducted our empirical analysis only for $m=4$.
In addition, for both sets of basket options, the Monte Carlo method with control variate detailed in \cite{Pellizzari1998} is employed as a benchmark. The number of simulations used are between $10^5$ and $4\cdot 10^6$, depending on the number of assets considered.

The results in relation to  Set 1 are summarized in Tables \ref{tab:randomOptions2-10}, \ref{tab:randomOptions11-15}, \ref{tab:randomOptions16-20} and \ref{tab:randomOptions21-50}, grouped for scenarios with the number of assets between 2-10, 11-15, 16-20 and 21-50 respectively, while Table \ref{tab:randomOptionsTotal} summarizes the results for all the 1000 instances in this set.
Overall, the methods $4GA$ and $4GB$ give the same results in terms of $RMSE$ (\CThree), with 4GA slightly better than $4GB$ for short maturities. Considering the comparison criterion \CTwo, the method $4GB$ is much better than the others. For small maturities $BPW$ performs slightly better than our method but the error associated with the $BPW$ method is at least double for all the other comparison criteria. Finally, considering \COne$\,$  both methods $4GA$  and $4GB$ perform much better than $BPW$.
 When applying the $BPW$ method, increasing the number of assets in the basket has the effect of increasing the $RMSE$  on the matched options. $BPW$'s performance is almost constant across  different categories, the only  slight improvement (1.5 decimal points on average) can be noticed at small maturities. Moreover, \CTwo, the percentage of non-matched options, increases with the number of assets, performing better for lower interest rates, short maturities and at-the-money options. Considering \COne, the number of minimum errors, the best results are obtained for a number of assets in the basket between 11 and 20. The $BPW$ method performs well for longer maturities. For $4GA$  there does not seem to be an explicit relation between number of assets and \CThree$\,$. However, our empirical results show that \COne$\,$ and \CTwo$\,$  decreases, increases respectively,  with the number of assets. Overall one can conclude that both Hermite approximation methods $4GA$ and $4GB$ have an excellent performance on large baskets.

\begin{center}
[Table \ref{tab:randomOptions2-10} about here.]
\end{center}
\begin{center}
[Table \ref{tab:randomOptions11-15} about here.]
\end{center}
\begin{center}
[Table \ref{tab:randomOptions16-20} about here.]
\end{center}
\begin{center}
[Table \ref{tab:randomOptions21-50} about here.]
\end{center}
\begin{center}
[Table \ref{tab:randomOptionsTotal} about here.]
\end{center}
Table   \ref{tab:groupBOverall} summarizes the results for  Set 2, reflecting the challenges  posed by taking into consideration the intensity of the Poisson processes. For the analysis in this group, we also consider a hybrid method spanned by the two methods $4GA$ and $4GB$, which will be called  $4GAB$ henceforth. This hybrid method $4GAB$ returns the solution of the method that matches correctly the moments if only one of $4GA$ and $4GB$ works properly. The comparison is carried considering the error of the method that matches the first four moments if only one of  $4GA$ and $4GB$ finds a solution, or the worst error if both find a numerical solution. Even though the method $4GAB$ considers the worst error between the A and B variants, it is superior to the other compared methods, being able to match the required basket moments in 96.6\% (1-\CTwo) of cases and reaching the minimum error (\COne) 84\% of the times.

\begin{center}
[Table \ref{tab:groupBOverall} about here.]
\end{center}

	\subsection{Delta-hedging performances}
\label{HedgingComparison}
A comparison of dynamic Delta-hedging performance between our formula~\eqref{DeltaGreek} and the formula proposed in $BPW$ (see definition of $\Delta_i$ in that paper\footnote{ \cite{BorovkovaPermanaWeide2007} report the formula for the sensitivity of the option with respect to individual stock prices in the basket. The sensibility with respect to $B_0$ can be calculated by multiplicating that formula for $\frac{\partial S_i}{\partial B_0}=\frac{1}{a_i}$.}) is illustrated in this section.
      A sample of 1000 simulated  paths, indexed by $s=1,\cdots,1000$, with 1-month-interval hedging rolling frequency are generated for the six basket options. The basket options considered are mostly those in Table~\ref{tab:baskets} with some modifications in order to have a more meaningful comparison. In the following, the options' characteristics are detailed:
\begin{itemize}
\item baskets $1^*$ and $2^*$ are exactly the same as  1 and 2 in Table~\ref{tab:baskets};
\item basket $3^*$ is equal to basket 3 but $\shift_0^{(i)}=10\cdot i \ee^{-rT}$, $\JIntensityQ_i=0.3$, $\logJumpMeanQ_i=-0.3$ and $\logJumpVolatilityQ_i=0.2$ for all $i=1,\cdots,\numAssets$;
\item basket $4^*$ is equal to basket 4 but $\shift_0^{(1)}=0$ and $\shift_0^{(2)}=50 \ee^{-rT}$, $\sign_1=-\sign_2=1$, $\JIntensityQ_i=0.3$, $\logJumpMeanQ_1=0.3$, $\logJumpMeanQ_2=0.1$ and $\logJumpVolatilityQ_i=0.2$ for all $i=1,\cdots,\numAssets$;
\item baskets $5^*$ and $6^*$ are respectively equal to basket 5 and 6 but $\shift_0^{(i)}=10\cdot i \ee^{-rT}$, $\JIntensityQ_i=0$.
\end{itemize}
 For each path, the option price and the option Delta are calculated at each time step. The evaluation of the performance for the Delta-hedged portfolios is performed via three different measures:
\begin{description}
\item[\CFour] \textit{average volatility of the Delta}, defined as:
\begin{equation}
\mbox{\CFour}^l=\frac{1}{1000}\sum_{i=1}^{1000}{\sigma_i^l}
\end{equation}
where $\sigma_i^l$ is the volatility of the Delta calculated by method $l$ along path $i$. A pricing method implying less volatile Delta is  better because hedging costs do not put liquidity pressure on the investor;
\item[\CFive] \textit{Square root of MSE} in the hedged portfolio evaluated (per month) as:
\begin{equation}
\mbox{\CFive}^l=\frac{1}{12 c_{t_0}}\sum_{i=0}^{n-1}{\left[c_{t_i}-c_{t_{i+1}}+\Delta^l_{t_i}(B_{t_i}-B_{t_{i+1}})\right]^2}
\end{equation}
where $c_{t_i}$ is the Monte Carlo price at time $t_i=\frac{T}{12}\cdot i$, $n=12$ and $\Delta^l_{t_i}$ is the Delta calculated at time $t_i$ by method $l$.
\item[\CSix-\CTen]\textit{ Ability of the hedging strategy to match the option value at maturity}. Outside transaction costs, we evaluate how far from zero is the value of the hedged portfolio at maturity $T$.  At time 0, the hedged-portfolio contains a short position in a call option, $\Delta_0$ position in the basket and cash in a money account that renders a null value for the portfolio at time 0. At each time step, the number of positions in the basket is changed according to $\Delta$ and consequently the money account. Five performance measures are used to evaluate the money-performance: the percentage of sub-hedging (\CSix), the percentage of super-hedging (\CSeven), the average error for the sub-hedged and super-hedged portfolios (\CEight$\,$ and \CNine$\,$ respectively), and the average error among all the simulations (\CTen).
\end{description}

The results for the hedging performance are reported in Table  \ref{tab:hedgingResults}.
The methods $4GA$ and $4GB$ produce very good results that are very similar with  $4GA$ only slightly better but this may be due to the particular simulations used in pricing options. However, both Hermite approximations methods are  superior to  the $BPW$ method for all measures of performance except the $RMSE$ (\CFive). For $BPW$, \CSix$\,$ and \CSeven$\,$ are almost the same, showing that this method may lead to under-hedging but also over-hedging.  At the same time the methods $4GA$ and $4GB$ seem to be occasionally only under-hedged, but the hedging error is small as indicated by \CTen.

\begin{center}
[Table \ref{tab:hedgingResults} about here.]
\end{center}

\section{Conclusions}
By introducing a shift parameter into the drift of the diffusion process underlying the assets of a basket, one can account for the empirical characteristics of historical prices of those assets. In particular, the modelling is laid on improved foundations, being able to cover the well-documented negative skewness. However, recent techniques imposed strong assumptions on the evolution dynamics of the basket as whole, searching for closed-form solution and repackaging of log-normal Black-Scholes type pricing formulae.

In this paper, we have shown that this path is not necessary and we have highlighted a methodology that  may work well with other future models in this area. We focused here on the shifted jump-diffusion model and we demonstrated with empirical simulations, that our Hermite expansion approach may provide pricing results that are as good as competing methods, and in many cases superior. In addition, we followed the hedging performance as a comparison tool and again our technology provided excellent results.

In our opinion, the improved results emphasized in the paper are not surprising since the technique is fundamentally based on matching the first four moments under model specification. Thus, we allow granular specification of dynamics for each asset and then only determine the moments of the basket. While our paper was focused on equity baskets, it is clear that the same methodology can be applied for mixtures of assets and models, as long as moments can be calculated easily.


\def\bibsection{\section*{References}}
\bibliographystyle{apalike}
\bibliography{bibliografynow}

\begin{table}[htbp]
\centering
\caption{Summary of the variants of the Hermite method considered in this work.\\ \footnotesize
The first column contains the names of the variants considered: $m$ stands for the number of moments matched, $G$ highlights that a transformation of the Gaussian distribution is considered (as shown by the variable $Z$ in the second column) and $A$ and $B$ identify the standardized returns used as approximated random variable (last column). In particular, two standardized returns are considered: variant $A$ is constructed in such a way that the first moment of the approximated random variable is 1 while the return in variant $B$ has first moment equal to 0. $B_t$ is defined in \eqref{shiftedBasket} as shifted basket, the moments of $J(Z)$ and  the moments of the quantities in the last column are in Appendix \ref{momentNormalizedBasket}, $H_k(x)$ denotes the $k$th-order Hermite polynomial $H_k(x)=\frac{(-1)^k}{\phi(x)}\frac{\partial^k \phi(x)}{\partial x^k}$ where $\phi(\cdot)$ is the standard normal density function and $\coeffHermite{k}$ are determined to exactly match the first $m$ moments of the approximated random variable (last column).}
\begin{tabular}{c|c|c}
\hline
Variant's name                 & \textit{Approximating} r.v. & \textit{Approximated} r.v. \\
\hline \tstrut
$mGA$ &     \multirow{2}{*}{$J(Z)=\sum_{k=0}^{m-1}{\coeffHermite{k} H_k(Z)}$ }        &  $ \frac{\ShiftedBasket_T}{\ShiftedBasket_0\ee^{rT}}$        \\[1ex]
$mGB$ &           & $ \frac{\ShiftedBasket_T}{ \ShiftedBasket_0\ee^{rT}}-1$           \\[1ex]
\hline
\end{tabular}
\label{momentMatching}
\end{table}

\begin{table}[htbp]
\centering
 \footnotesize
 \caption{Specification of the basket options under multi-dimensional GBM model.\\ \footnotesize This specification follows  \cite{BorovkovaPermanaWeide2007}. Other relevant parameters are risk-free rate equal to 3\%, 1-year maturity, $\JIntensityQ=0$, $\shift_0^{(i)}=0$ and $\sign_i=1$. The first row indicates $[S_0^{(1)},S_0^{(2)},S_0^{(3)}]$, the second $[\sigma_1,\sigma_2,\sigma_3]$, the third $[a_1,a_2,a_3]$, the forth the correlation $\rho_{i,j}$ for each couple $(i,j)$ of assets and the fifth $\Strike$. The only difference with the options in \cite{BorovkovaPermanaWeide2007} is that they price options on basket of forward contracts while we price options on basket of assets.}
\begin{tabular}{|r|c|c|c|c|c|c|}
\cline{2-7}    \multicolumn{1}{r|}{} & Basket 1 & Basket 2 & Basket 3 & Basket 4 & Basket 5 & Basket 6 \\
\hline
Stock Prices & [100,120] & [150,100] & [110,90] & [200,50] & [95,90,105] & [100,90,95] \\
    \hline
    Volatility & [0.2,0.3] & [0.3,0.2] & [0.3,0.2] & [0.1,0.15] & [0.2,0.3,0.25] & [0.25,0.3,0.2] \\
    \hline
    Weights & [-1,1] & [-1,1] & [0.7,0.3] & [-1,1] & [1,-0.8,-0.5] & [0.6,0.8,-1] \\
        \hline
        \multicolumn{1}{|c|}{\multirow{3}[2]{*}{Correlation}} & \multirow{3}[2]{*}{$\rho_{1,2}=0.9$} & \multirow{3}[2]{*}{$\rho_{1,2}=0.3$} & \multirow{3}[2]{*}{$\rho_{1,2}=0.9$} & \multirow{3}[2]{*}{$\rho_{1,2}=0.8$} & $\rho_{1,2}=0.9$,  & $\rho_{1,2}=0.9$,  \\
        \multicolumn{1}{|c|}{} &  &  &  &  & $\rho_{2,3}=0.9$ & $\rho_{2,3}=0.9$ \\
        \multicolumn{1}{|c|}{} &  &  &  &  & $\rho_{1,3}=0.8$ & $\rho_{1,3}=0.8$ \\
        \hline
        Strike price & 20 & -50 & 104 & -140 & -30 & 35 \\
            \hline
\end{tabular}
\label{tab:baskets}
\end{table}

\begin{table}[htbp]
\centering
\footnotesize
\caption{Comparison under multi-dimensional GBM model.\\ \footnotesize This table reports the comparison on the six basket options in \cite{BorovkovaPermanaWeide2007}. In the second column, the prices (standard deviation in bracket) calculated by the Monte Carlo method with control variate in \cite{Pellizzari1998} with $4\times 10^6$ simulations are reported and they are considered as benchmark. In the third column, there are the prices calculated by the method in \cite{BorovkovaPermanaWeide2007}. The last four columns contain the prices under the methods $mGA$ and $mGB$ when $m=4$ and $m=6$. Two of the measures of error considered are reported in the last two rows: \COne-- the percentage of times the minimum squared error is reached under the specified method, \CThree-- the square root of MSE calculated only relative to the options for which the method was able to find a numerical solution. The third measure of error, \CTwo, that indicates the percentage of times the relative error is greater than 5\%, is always equal to 0 and it is not reported in the table.}

\begin{tabular}{|c|c|r|r|r|r|r|}
\cline{2-7}    \multicolumn{1}{r|}{} & MC  & \multicolumn{1}{c|}{\multirow{2}[2]{*}{BPW}} & \multicolumn{1}{c|}{\multirow{2}[2]{*}{$4GA$}} & \multicolumn{1}{c|}{\multirow{2}[2]{*}{$4GB$}} & \multicolumn{1}{c|}{\multirow{2}[2]{*}{$6GA$}} & \multicolumn{1}{c|}{\multirow{2}[2]{*}{$6GB$}} \bigstrut[t]\\
\multicolumn{1}{r|}{} & (SD) & \multicolumn{1}{c|}{} & \multicolumn{1}{c|}{} & \multicolumn{1}{c|}{} & \multicolumn{1}{c|}{} & \multicolumn{1}{c|}{} \bigstrut[b]\\
    \hline
     \multicolumn{1}{|c|}{\multirow{2}[2]{*}{Basket 1}} & 8.2263 & \multicolumn{1}{c|}{\multirow{2}[2]{*}{8.2442}} & \multicolumn{1}{c|}{\multirow{2}[2]{*}{8.1977}} & \multicolumn{1}{c|}{\multirow{2}[2]{*}{8.1977}} & \multicolumn{1}{c|}{\multirow{2}[2]{*}{8.2222}} & \multicolumn{1}{c|}{\multirow{2}[2]{*}{8.2222}}   \bigstrut[t]\\
        \multicolumn{1}{|c|}{} & (0.0031) & \multicolumn{1}{c|}{} & \multicolumn{1}{c|}{} & \multicolumn{1}{c|}{} & \multicolumn{1}{c|}{} & \multicolumn{1}{c|}{} \bigstrut[b]\\
        \hline
        \multicolumn{1}{|c|}{\multirow{2}[2]{*}{Basket 2}} & 16.4700 & \multicolumn{1}{c|}{\multirow{2}[2]{*}{16.6215}} & \multicolumn{1}{c|}{\multirow{2}[2]{*}{16.4424}} & \multicolumn{1}{c|}{\multirow{2}[2]{*}{16.4424}} & \multicolumn{1}{c|}{\multirow{2}[2]{*}{16.4631}} & \multicolumn{1}{c|}{\multirow{2}[2]{*}{16.3654}}   \bigstrut[t]\\
        \multicolumn{1}{|c|}{} & (0.0052) & \multicolumn{1}{c|}{} & \multicolumn{1}{c|}{} & \multicolumn{1}{c|}{} & \multicolumn{1}{c|}{} & \multicolumn{1}{c|}{} \bigstrut[b]\\
        \hline
        \multicolumn{1}{|c|}{\multirow{2}[2]{*}{Basket 3}} & 12.5887 & \multicolumn{1}{c|}{\multirow{2}[2]{*}{12.5911}} & \multicolumn{1}{c|}{\multirow{2}[2]{*}{12.5695}} & \multicolumn{1}{c|}{\multirow{2}[2]{*}{12.5695}} & \multicolumn{1}{c|}{\multirow{2}[2]{*}{12.5888}} & \multicolumn{1}{c|}{\multirow{2}[2]{*}{12.5888}}  \bigstrut[t]\\
        \multicolumn{1}{|c|}{} & (0.0005) & \multicolumn{1}{c|}{} & \multicolumn{1}{c|}{} & \multicolumn{1}{c|}{} & \multicolumn{1}{c|}{} & \multicolumn{1}{c|}{} \bigstrut[b]\\
        \hline
        \multicolumn{1}{|c|}{\multirow{2}[2]{*}{Basket 4}} & 1.1459 & \multicolumn{1}{c|}{\multirow{2}[2]{*}{1.1456}} & \multicolumn{1}{c|}{\multirow{2}[2]{*}{1.1453}} & \multicolumn{1}{c|}{\multirow{2}[2]{*}{1.1453}} & \multicolumn{1}{c|}{\multirow{2}[2]{*}{1.0938}} & \multicolumn{1}{c|}{\multirow{2}[2]{*}{1.1162}}  \bigstrut[t]\\
        \multicolumn{1}{|c|}{} & (0.0008) & \multicolumn{1}{c|}{} & \multicolumn{1}{c|}{} & \multicolumn{1}{c|}{} & \multicolumn{1}{c|}{} & \multicolumn{1}{c|}{} \bigstrut[b]\\
        \hline
        \multicolumn{1}{|c|}{\multirow{2}[2]{*}{Basket 5}} & 7.4681 & \multicolumn{1}{c|}{\multirow{2}[2]{*}{7.4951}} & \multicolumn{1}{c|}{\multirow{2}[2]{*}{7.4563}} & \multicolumn{1}{c|}{\multirow{2}[2]{*}{7.4563}} & \multicolumn{1}{c|}{\multirow{2}[2]{*}{7.4555}} & \multicolumn{1}{c|}{\multirow{2}[2]{*}{7.4555}}  \bigstrut[t]\\
        \multicolumn{1}{|c|}{} & (0.0027) & \multicolumn{1}{c|}{} & \multicolumn{1}{c|}{} & \multicolumn{1}{c|}{} & \multicolumn{1}{c|}{} & \multicolumn{1}{c|}{} \bigstrut[b]\\
        \hline
        \multicolumn{1}{|c|}{\multirow{2}[2]{*}{Basket 6}} & 9.7767 & \multicolumn{1}{c|}{\multirow{2}[2]{*}{9.7989}} & \multicolumn{1}{c|}{\multirow{2}[2]{*}{9.7628}} & \multicolumn{1}{c|}{\multirow{2}[2]{*}{9.7628}} & \multicolumn{1}{c|}{\multirow{2}[2]{*}{9.7856}} & \multicolumn{1}{c|}{\multirow{2}[2]{*}{9.7856}}  \bigstrut[t]\\
        \multicolumn{1}{|c|}{} & (0.0030) & \multicolumn{1}{c|}{} & \multicolumn{1}{c|}{} & \multicolumn{1}{c|}{} & \multicolumn{1}{c|}{} & \multicolumn{1}{c|}{} \bigstrut[b]\\
        \hline
        \COne &  -  & 33.33\% & 16.67\% & 16.67\% & \textbf{66.67\%} & 33.33\%  \bigstrut\\
                \hline
        \CThree &  -  & 0.0635 & \textbf{0.0195} & \textbf{0.0195} & 0.0224 & 0.0454  \bigstrut\\
        \hline

\end{tabular}
\label{tab:GBMresults}
\end{table}

\begin{landscape}
\begin{table}[htbp]
  \centering
  \footnotesize
\caption{Comparison I (Set 1): number of assets between 2 and 10.\\
 \footnotesize
This table contains the summary of the performances of several methods for pricing options in Set 1 with numbers of assets randomly generated between 2 and 10. The assets follow equation \eqref{assetSolutionShifted} where the parameters are randomly generated and uniformly distributed in the following ranges: $r\in[0;0.1]$, $\sigma_i\in[0.1;0.6]$, $T\in [0.1;1]$, $S_0^{(i)}=[70;130]$, $a_i\in[-1;1]$, $\frac{\Strike}{\Basket}\in[0.95;1.05]$, $\shift_0^{(i)}\ee^{rT}\in[-20;20]$, $\sign_i\in[-1;1]$, $\JIntensityQ_i\in[0;0.2]$, $\logJumpMeanQ_i\in[-0.3;0]$ and $\logJumpVolatilityQ_i\in[0;0.3]$. Three measures of error are reported: \COne-- the percentage of times the minimum squared error is reached under the specific method; \CTwo-- the  percentage of times the relative error is greater than 5\% for the specified method; \CThree-- the square root of MSE calculated only relative to the options for which the method was able to find a numerical solution.
The results are shown (per column) along three different dimensions: risk-free rate, time to maturity and strike price. Along the different rows, the results per  method are showed: in particular, $BPW$ stands for the method in \cite{BorovkovaPermanaWeide2007} and $mGA$ and $mGB$ are considered for both $m=4$ and $m=6$.}
\footnotesize
    \begin{tabular}{|r|r|r|r|r|r|r|r|r|r|r|}
\cline{3-9}    \multicolumn{1}{c}{} &       & \multicolumn{2}{c|}{$r$} & \multicolumn{2}{c|}{$T$} & \multicolumn{3}{c|}{$K$} & \multicolumn{1}{r}{} & \multicolumn{1}{r}{} \bigstrut\\
\cline{3-10}    \multicolumn{1}{r}{} &       & \multicolumn{1}{c|}{$r\leq 0.05$} & \multicolumn{1}{c|}{$r>0.05$} & \multicolumn{1}{c|}{$T\leq 0.5$} & \multicolumn{1}{c|}{$T>0.5$} & \multicolumn{1}{c|}{$\frac{K}{\Basket_0}\leq 0.98$} & \multicolumn{1}{c|}{$0.98< \frac{K}{\Basket_0} \leq 1.02$} & \multicolumn{1}{c|}{$\frac{K}{\Basket_0}>1.02$} & \multicolumn{1}{c|}{Total} & \multicolumn{1}{r}{} \bigstrut\\
    \hline
    \multicolumn{1}{|c|}{\multirow{5}[9]{*}{\COne}} &  $6GA$   & 42.54\% & \textbf{48.71\%} & 39.68\% & \textbf{50.99\%} & \textbf{47.18\%} & 44.95\% & 44.38\% & 45.40\% &  $6GA$ \bigstrut\\
\cline{2-11}    \multicolumn{1}{|c|}{} &  $6GB$   & 38.43\% & 41.81\% & 32.79\% & 47.04\% & 40.14\% & 42.93\% & 36.25\% & 40.00\% &  $6GB$ \bigstrut\\
\cline{2-11}    \multicolumn{1}{|c|}{} & BPW   & 15.30\% & 12.07\% & 20.24\% & 7.51\% & 17.61\% & 13.13\% & 11.25\% & 13.80\% & BPW \bigstrut\\
\cline{2-11}    \multicolumn{1}{|c|}{} &  $4GA$   & 46.64\% & 43.97\% & 51.82\% & 39.13\% & 44.37\% & 42.42\% & 50.00\% & 45.40\% &  $4GA$ \bigstrut\\
\cline{2-11}    \multicolumn{1}{|c|}{} &  $4GB$   & \textbf{50.00\%} & 46.55\% & \textbf{56.68\%} & 40.32\% & 43.66\% & \textbf{49.49\%} & \textbf{51.25\%} & \textbf{48.40\%} &  $4GB$ \bigstrut[t]\\
\hline
    \multicolumn{1}{|c|}{\multirow{5}[10]{*}{\CTwo}} &  $6GA$   & \textbf{0.00\%} & \textbf{0.86\%} & \textbf{0.40\%} & \textbf{0.40\%} & \textbf{0.00\%} & \textbf{1.01\%} & \textbf{0.00\%} & \textbf{0.40\%} &  $6GA$ \bigstrut\\
\cline{2-11}    \multicolumn{1}{|c|}{} &  $6GB$   & 1.12\% & 1.72\% & 2.43\% & \textbf{0.40\%} & 0.70\% & 2.53\% & 0.63\% & 1.40\% &  $6GB$ \bigstrut\\
\cline{2-11}    \multicolumn{1}{|c|}{} & BPW   & 11.57\% & 8.19\% & 2.43\% & 17.39\% & 9.15\% & 9.09\% & 11.88\% & 10.00\% & BPW \bigstrut\\
\cline{2-11}    \multicolumn{1}{|c|}{} &  $4GA$   & 9.33\% & 10.34\% & 16.19\% & 3.56\% & 8.45\% & 13.13\% & 6.88\% & 9.80\% &  $4GA$ \bigstrut\\
\cline{2-11}    \multicolumn{1}{|c|}{} &  $4GB$   & 5.60\% & 3.88\% & 6.48\% & 3.16\% & 3.52\% & 5.56\% & 5.00\% & 4.80\% &  $4GB$ \bigstrut\\
    \hline
 \multicolumn{1}{|c|}{\multirow{5}[10]{*}{\CThree}}           &  $6GA$   & 0.1473 & 0.1632 & 0.1057 & 0.191 & 0.1515 & 0.1719 & 0.1343 & 0.1549 &  $6GA$ \bigstrut\\
 \cline{2-11}    \multicolumn{1}{|c|}{}  &  $6GB$   & 0.1356 & 0.1517 & 0.0977 & 0.1768 & 0.1475 & 0.1501 & 0.1304 & 0.1433 &  $6GB$ \bigstrut\\
\cline{2-11}    \multicolumn{1}{|c|}{} & BPW   & 0.2528 & 0.2627 & 0.1553 & 0.3278 & 0.2378 & 0.2608 & 0.2696 & 0.2574 & BPW \bigstrut\\
\cline{2-11}    \multicolumn{1}{|c|}{}          &  $4GA$   & \textbf{0.1159} & \textbf{0.1183} & \textbf{0.036} & 0.1607 & \textbf{0.1243} & \textbf{0.1079} & \textbf{0.1212} & \textbf{0.1171} &  $4GA$ \bigstrut\\
\cline{2-11}    \multicolumn{1}{|c|}{}          &  $4GB$   & 0.13  & 0.1231 & 0.0813 & \textbf{0.1592} & 0.1275 & 0.1139 & 0.1406 & 0.1268 &  $4GB$ \bigstrut\\
    \hline
    \multicolumn{2}{|c|}{\# options} & 268   & 232   & 247   & 253   & 142   & 198   & 160   & 500   & \multicolumn{1}{r}{} \bigstrut[b]\\
\cline{1-10}    \end{tabular}%
  \label{tab:randomOptions2-10}%
\end{table}
\end{landscape}%

\begin{landscape}
\begin{table}[htbp]
  \centering
   \footnotesize
\caption{Comparison II (Set 1): number of assets between 11 and 15.\\\footnotesize
This table contains the summary of the performances of several methods for pricing options in Set 1 with numbers of assets randomly generated between 11 and 15. The assets follow equation \eqref{assetSolutionShifted}. The assets follow equation \eqref{assetSolutionShifted} where the parameters are randomly generated and uniformly distributed in the following ranges: $r\in[0;0.1]$, $\sigma_i\in[0.1;0.6]$, $T\in [0.1;1]$, $S_0^{(i)}=[70;130]$, $a_i\in[-1;1]$, $\frac{\Strike}{\Basket}\in[0.95;1.05]$, $\shift_0^{(i)}\ee^{rT}\in[-20;20]$, $\sign_i\in[-1;1]$, $\JIntensityQ_i\in[0;0.2]$, $\logJumpMeanQ_i\in[-0.3;0]$ and $\logJumpVolatilityQ_i\in[0;0.3]$. } 
    \begin{tabular}{|r|r|r|r|r|r|r|r|r|r|r|}
\cline{3-9}    \multicolumn{1}{r}{} &       & \multicolumn{2}{c|}{$r$} & \multicolumn{2}{c|}{$T$} & \multicolumn{3}{c|}{$K$} & \multicolumn{1}{r}{} & \multicolumn{1}{r}{} \bigstrut\\
\cline{3-10}    \multicolumn{1}{r}{} &       & \multicolumn{1}{c|}{$r\leq 0.05$} & \multicolumn{1}{c|}{$r>0.05$} & \multicolumn{1}{c|}{$T\leq 0.5$} & \multicolumn{1}{c|}{$T>0.5$} & \multicolumn{1}{c|}{$\frac{K}{\Basket_0}\leq 0.98$} & \multicolumn{1}{c|}{$0.98 < \frac{K}{\Basket_0}\leq 1.02$} & \multicolumn{1}{c|}{$\frac{K}{\Basket_0} >1.02 $} & \multicolumn{1}{c|}{Total} & \multicolumn{1}{r}{} \bigstrut\\
    \hline
\multicolumn{1}{|c|}{\multirow{3}[6]{*}{\COne}} & BPW   & 11.89\% & 15.29\% & 20.57\% & 7.55\% & 13.83\% & 16.10\% & 10.23\% & 13.67\% & BPW \bigstrut\\
\cline{2-11}    \multicolumn{1}{|c|}{} &  $4GA$   & 81.82\% & 75.80\% & 66.67\% & 89.31\% & 82.98\% & 72.88\% & 81.82\% & 78.67\% &  $4GA$ \bigstrut\\
\cline{2-11}    \multicolumn{1}{|c|}{} &  $4GB$   & \textbf{86.01\%} & \textbf{84.08\%} & \textbf{78.01\%} & \textbf{91.19\%} & \textbf{85.11\%} & \textbf{84.75\%} & \textbf{85.23\%} &\textbf{ 85.00\%} &  $4GB$ \bigstrut\\
    \hline
    \multicolumn{1}{|c|}{\multirow{3}[6]{*}{\CTwo}} & BPW   & 12.59\% & 15.92\% & \textbf{3.55\% }& 23.90\% & 14.89\% & 11.02\% & 18.18\% & 14.33\% & BPW \bigstrut\\
\cline{2-11}    \multicolumn{1}{|c|}{} &  $4GA$   & 11.19\% & 14.01\% & 22.70\% & 3.77\% & 7.45\% & 18.64\% & 10.23\% & 12.67\% &  $4GA$ \bigstrut\\
\cline{2-11}    \multicolumn{1}{|c|}{} &  $4GB$   & \textbf{5.59\%} & \textbf{3.82\%} & 7.80\% & \textbf{1.89\%} & \textbf{5.32\%} & \textbf{3.39\%} & \textbf{5.68\%} & \textbf{4.67\%} &  $4GB$ \bigstrut\\
    \hline
    \multicolumn{1}{|c|}{\multirow{3}[6]{*}{\CThree}} & BPW   & 0.3703 & 0.3778 & 0.1922 & 0.4811 & 0.3718 & 0.3882 & 0.3573 & 0.3742 & BPW \bigstrut\\
\cline{2-11}    \multicolumn{1}{|c|}{} &  $4GA$   & \textbf{0.1025} & \textbf{0.1401} & \textbf{0.0385} & \textbf{0.1659} & \textbf{0.1266} & \textbf{0.1001} & \textbf{0.1468} & \textbf{0.1236} &  $4GA$ \bigstrut\\
\cline{2-11}    \multicolumn{1}{|c|}{}          &  $4GB$   & 0.1138 & 0.1556 & 0.0637 & 0.1788 & \textbf{0.1266} & 0.1145 & 0.1719 & 0.1373 &  $4GB$ \bigstrut\\
    \hline
    \multicolumn{2}{|c|}{\# options} & 143   & 157   & 141   & 159   & 94    & 118   & 88    & 300   & \multicolumn{1}{r}{} \bigstrut\\
\cline{1-10}    \end{tabular}%
\label{tab:randomOptions11-15}%
\end{table}
\end{landscape}%

\begin{landscape}
\begin{table}[htbp]
  \centering
   \footnotesize
   \caption{Comparison III (Set 1): number of assets between 16 and 20.\\
   \footnotesize
   This table contains the summary of the performances of several methods for pricing options in Set 1 with numbers of assets randomly generated between 16 and 20.  The assets follow equation \eqref{assetSolutionShifted}. The assets follow equation \eqref{assetSolutionShifted} where the parameters are randomly generated and uniformly distributed in the following ranges: $r\in[0;0.1]$, $\sigma_i\in[0.1;0.6]$, $T\in [0.1;1]$, $S_0^{(i)}=[70;130]$, $a_i\in[-1;1]$, $\frac{\Strike}{\Basket}\in[0.95;1.05]$, $\shift_0^{(i)}\ee^{rT}\in[-20;20]$, $\sign_i\in[-1;1]$, $\JIntensityQ_i\in[0;0.2]$, $\logJumpMeanQ_i\in[-0.3;0]$ and $\logJumpVolatilityQ_i\in[0;0.3]$. } 
    \begin{tabular}{|r|r|r|r|r|r|r|r|r|r|r|}
\cline{3-9}    \multicolumn{1}{r}{} &       & \multicolumn{2}{c|}{$r$} & \multicolumn{2}{c|}{$T$} & \multicolumn{3}{c|}{$K$} & \multicolumn{1}{r}{} & \multicolumn{1}{r}{} \bigstrut\\
\cline{3-10}    \multicolumn{1}{r}{} &       & \multicolumn{1}{c|}{$r\leq 0.05$} & \multicolumn{1}{c|}{$r>0.05$} & \multicolumn{1}{c|}{$T\leq 0.5$} & \multicolumn{1}{c|}{$T>0.5$} & \multicolumn{1}{c|}{$\frac{K}{\Basket_0}\leq 0.98$} & \multicolumn{1}{c|}{$0.98 < \frac{K}{\Basket_0}\leq 1.02$} & \multicolumn{1}{c|}{$\frac{K}{\Basket_0} >1.02 $} & \multicolumn{1}{c|}{Total} & \multicolumn{1}{r}{} \bigstrut\\
    \hline
  \multicolumn{1}{|c|}{\multirow{3}[6]{*}{\COne}} & BPW   & 12.50\% & 15.38\% & 25.00\% & 5.36\% & 12.50\% & 13.79\% & 15.38\% & 14.00\% & BPW \bigstrut\\
\cline{2-11}    \multicolumn{1}{|c|}{} &  $4GA$   & \textbf{85.42\%} & \textbf{82.69\%} & 68.18\% & \textbf{96.43\%} & \textbf{87.50\%} & 82.76\% & \textbf{82.05\% }& \textbf{84.00\%} &  $4GA$ \bigstrut\\
\cline{2-11}    \multicolumn{1}{|c|}{} &  $4GB$   & 83.33\% & \textbf{82.69\% }& \textbf{70.45\%} & 92.86\% & \textbf{87.50\%} & \textbf{86.21\% }& 76.92\% & 83.00\% &  $4GB$ \bigstrut\\
    \hline
    \multicolumn{1}{|c|}{\multirow{3}[6]{*}{\CTwo}} & BPW   & 14.58\% & 23.08\% & 4.55\% & 30.36\% & 31.25\% & 6.90\% & 17.95\% & 19.00\% & BPW \bigstrut\\
\cline{2-11}    \multicolumn{1}{|c|}{} &  $4GA$   &\textbf{ 2.08\% }& \textbf{7.69\%} & 11.36\% & \textbf{0.00\%} & 6.25\% & 3.45\% &\textbf{ 5.13\%} & \textbf{5.00\%} &  $4GA$ \bigstrut\\
\cline{2-11}    \multicolumn{1}{|c|}{} &  $4GB$   &\textbf{ 2.08\%} & \textbf{7.69\%} & \textbf{6.82\% }& 3.57\% & \textbf{3.13\%} & \textbf{0.00\%} & 10.26\% & \textbf{5.00\%} &  $4GB$ \bigstrut\\
    \hline
    \multicolumn{1}{|c|}{\multirow{3}[6]{*}{\CThree}} & BPW   & 0.2717 & 0.3638 & 0.2135 & 0.3877 & 0.2974 & 0.3462 & 0.325 & 0.3229 & BPW \bigstrut\\
\cline{2-11}    \multicolumn{1}{|c|}{} &  $4GA$   & \textbf{0.1136} & 0.1186 & \textbf{0.0521} & 0.1483 & \textbf{0.1466} & \textbf{0.0886} & 0.1056 & \textbf{0.1162} &  $4GA$ \bigstrut\\
\cline{2-11}    \multicolumn{1}{|c|}{}      &  $4GB$   & 0.1621 &\textbf{ 0.1166} & 0.1303 &\textbf{ 0.1476} & 0.2039 & 0.0887 & \textbf{0.1023} & 0.1403 &  $4GB$ \bigstrut\\
    \hline
    \multicolumn{2}{|c|}{\# options} & 48    & 52    & 44    & 56    & 32    & 29    & 39    & 100   & \multicolumn{1}{r}{} \bigstrut\\
\cline{1-10}    \end{tabular}%
\label{tab:randomOptions16-20}%
\end{table}
\end{landscape}%

\begin{landscape}
\begin{table}[htbp]
  \centering
   \footnotesize

   \caption{Comparison IV (Set 1): number of assets between 31 and 50.\\ \footnotesize
This table contains the summary of the performances of several methods for pricing options in Set 1 with numbers of assets randomly generated between 31 and 50.  The assets follow equation \eqref{assetSolutionShifted}. The assets follow equation \eqref{assetSolutionShifted} where the parameters are randomly generated and uniformly distributed in the following ranges: $r\in[0;0.1]$, $\sigma_i\in[0.1;0.6]$, $T\in [0.1;1]$, $S_0^{(i)}=[70;130]$, $a_i\in[-1;1]$, $\frac{\Strike}{\Basket}\in[0.95;1.05]$, $\shift_0^{(i)}\ee^{rT}\in[-20;20]$, $\sign_i\in[-1;1]$, $\JIntensityQ_i\in[0;0.2]$, $\logJumpMeanQ_i\in[-0.3;0]$ and $\logJumpVolatilityQ_i\in[0;0.3]$. } 
    \begin{tabular}{|r|r|r|r|r|r|r|r|r|r|r|}
\cline{3-9}    \multicolumn{1}{r}{} &       & \multicolumn{2}{c|}{$r$} & \multicolumn{2}{c|}{$T$} & \multicolumn{3}{c|}{$K$} & \multicolumn{1}{r}{} & \multicolumn{1}{r}{} \bigstrut\\
\cline{3-10}    \multicolumn{1}{r}{} &       & \multicolumn{1}{c|}{$r\leq 0.05$} & \multicolumn{1}{c|}{$r>0.05$} & \multicolumn{1}{c|}{$T\leq 0.5$} & \multicolumn{1}{c|}{$T>0.5$} & \multicolumn{1}{c|}{$\frac{K}{\Basket_0}\leq 0.98$} & \multicolumn{1}{c|}{$0.98 < \frac{K}{\Basket_0}\leq 1.02$} & \multicolumn{1}{c|}{$\frac{K}{\Basket_0} >1.02 $} & \multicolumn{1}{c|}{Total} & \multicolumn{1}{r}{} \bigstrut\\
    \hline
 \multicolumn{1}{|c|}{\multirow{3}[6]{*}{\COne}} & BPW   & 12.24\% & 5.88\% & 15.38\% & 2.08\% & 14.81\% & 2.38\% & 12.90\% & 9.00\% & BPW \bigstrut\\
\cline{2-11}    \multicolumn{1}{|c|}{} &  $4GA$   & \textbf{85.71\%} & 90.20\% & 78.85\% & \textbf{97.92\%} & \textbf{85.19\%} & 92.86\% & 83.87\% & 88.00\% &  $4GA$ \bigstrut\\
\cline{2-11}    \multicolumn{1}{|c|}{} &  $4GB$   & \textbf{85.71\%} &\textbf{ 94.12\%} & \textbf{82.69\%} & \textbf{97.92\%} & 81.48\% & \textbf{97.62\%} & \textbf{87.10\%} & \textbf{90.00\%} &  $4GB$ \bigstrut\\
    \hline
    \multicolumn{1}{|c|}{\multirow{3}[6]{*}{\CTwo}} & BPW   & 16.33\% & 35.29\% & 13.46\% & 39.58\% & 29.63\% & 19.05\% & 32.26\% & 26.00\% & BPW \bigstrut\\
\cline{2-11}    \multicolumn{1}{|c|}{} &  $4GA$   &\textbf{ 2.04\%} & 3.92\% & 5.77\% & \textbf{0.00\%} & \textbf{3.70\%} & 4.76\% & \textbf{0.00\%} & 3.00\% &  $4GA$ \bigstrut\\
\cline{2-11}    \multicolumn{1}{|c|}{} &  $4GB$   & \textbf{2.04\%} &\textbf{ 0.00\%} & \textbf{1.92\%} & \textbf{0.00\%} & \textbf{3.70\%} & \textbf{0.00\%} & \textbf{0.00\%} & \textbf{1.00\%} &  $4GB$ \bigstrut\\
    \hline
    \multicolumn{1}{|c|}{\multirow{3}[6]{*}{\CThree}} & BPW   & 0.2973 & 0.373 & 0.2742 & 0.3957 & 0.3311 & 0.3751 & 0.2872 & 0.338 & BPW \bigstrut\\
\cline{2-11}    \multicolumn{1}{|c|}{} &  $4GA$   & 0.1075 & \textbf{0.122 }& 0.0784 & \textbf{0.1447} & 0.1025 & \textbf{0.0732} &\textbf{ 0.1622} & 0.1151 &  $4GA$ \bigstrut\\
\cline{2-11}    \multicolumn{1}{|c|}{}           &  $4GB$   & \textbf{0.1072} & \textbf{0.122} & \textbf{0.0782} & \textbf{0.1447} & \textbf{0.102} & \textbf{0.0734 }& \textbf{0.1622} & \textbf{0.115} &  $4GB$ \bigstrut\\
    \hline
    \multicolumn{2}{|c|}{\# options} & 49    & 51    & 52    & 48    & 27    & 42    & 31    & 100   & \multicolumn{1}{r}{} \bigstrut\\
\cline{1-10}    \end{tabular}%
\label{tab:randomOptions21-50}%
\end{table}
\end{landscape}%

\begin{landscape}
\begin{table}[htbp]
  \centering
  \footnotesize
  \caption{Comparison V (Set 1): Total summary.\\ \footnotesize
  This table contains the summary of the performances of several methods for pricing options in Set 1.  The assets follow equation \eqref{assetSolutionShifted}. The assets follow equation \eqref{assetSolutionShifted} where the parameters are randomly generated and uniformly distributed in the following ranges: $r\in[0;0.1]$, $\sigma_i\in[0.1;0.6]$, $T\in [0.1;1]$, $S_0^{(i)}=[70;130]$, $a_i\in[-1;1]$, $\frac{\Strike}{\Basket}\in[0.95;1.05]$, $\shift_0^{(i)}\ee^{rT}\in[-20;20]$, $\sign_i\in[-1;1]$, $\JIntensityQ_i\in[0;0.2]$, $\logJumpMeanQ_i\in[-0.3;0]$ and $\logJumpVolatilityQ_i\in[0;0.3]$. } 

    \begin{tabular}{|r|r|r|r|r|r|r|r|r|r|r|}
\cline{3-9}    \multicolumn{1}{r}{} &  & \multicolumn{2}{c|}{$r$} & \multicolumn{2}{c|}{$T$} & \multicolumn{3}{c|}{$K$} & \multicolumn{1}{r}{} & \multicolumn{1}{r}{} \bigstrut\\
\cline{3-10}    \multicolumn{1}{r}{} &  & \multicolumn{1}{c|}{$r\leq 0.05$} & \multicolumn{1}{c|}{$r>0.05$} & \multicolumn{1}{c|}{$T\leq 0.5$} & \multicolumn{1}{c|}{$T>0.5$} & \multicolumn{1}{c|}{$\frac{K}{\Basket_0}\leq 0.98$} & \multicolumn{1}{c|}{$0.98 < \frac{K}{\Basket_0}\leq 1.02$} & \multicolumn{1}{c|}{$\frac{K}{\Basket_0} >1.02 $} & \multicolumn{1}{c|}{Total} & \multicolumn{1}{r}{} \bigstrut\\
    \hline
    \multicolumn{1}{|c|}{\multirow{3}[6]{*}{\COne}} & BPW & 13.78\% & 12.80\% & 20.25\% & 6.78\% & 15.59\% & 12.92\% & 11.64\% & 13.30\% & BPW \bigstrut\\
\cline{2-11}    \multicolumn{1}{|c|}{} &  $4GA$ & 63.98\% & 63.01\% & 60.54\% & 66.28\% & \textbf{65.08\%} & 60.21\% & 66.04\% & 63.50\% &  $4GA$ \bigstrut\\
\cline{2-11}    \multicolumn{1}{|c|}{} &  $4GB$ & \textbf{66.73\%} & \textbf{67.28\% }& \textbf{66.94\%} & \textbf{67.05\%} & \textbf{65.08\%} & \textbf{68.22\%} & \textbf{67.30\%} & \textbf{67.00\%} &  $4GB$ \bigstrut\\
    \hline
    \multicolumn{1}{|c|}{\multirow{3}[6]{*}{\CTwo}} & BPW & 12.60\% & 15.04\% & \textbf{4.13\%} & 22.87\% & 15.25\% & 10.59\% & 16.35\% & 13.80\% & BPW \bigstrut\\
\cline{2-11}    \multicolumn{1}{|c|}{} &  $4GA$ & 8.46\% & 10.57\% & 16.53\% & 2.91\% & 7.46\% & 13.18\% & 6.92\% & 9.50\% &  $4GA$ \bigstrut\\
\cline{2-11}    \multicolumn{1}{|c|}{} &  $4GB$ & \textbf{4.92\%} & \textbf{3.86\%} & 6.40\% &\textbf{ 2.52\%} & \textbf{4.07\%} & \textbf{3.88\%} & \textbf{5.35\%} & \textbf{4.40\%} &  $4GB$ \bigstrut\\
    \hline
 \multicolumn{1}{|c|}{\multirow{3}[6]{*}{\CThree}}& BPW & 0.332 & 0.3576 & 0.2348 & 0.4216 & 0.3342 & 0.3588 & 0.3375 & 0.3455 & BPW \bigstrut\\
\cline{2-11}    \multicolumn{1}{|c|}{} &  $4GA$ & \textbf{0.148 }& \textbf{0.1628} & \textbf{0.0959} & \textbf{0.1945} & \textbf{0.1579} & \textbf{0.1301} & \textbf{0.1765 }& \textbf{0.1559} &  $4GA$ \bigstrut\\
\cline{2-11}    \multicolumn{1}{|c|}{}   &  $4GB$ & 0.1603 & 0.1685 & 0.1185 & 0.1973 & 0.1652 & 0.1362 & 0.1898 & 0.1648 &  $4GB$ \bigstrut\\
    \hline
    \multicolumn{2}{|c|}{\# options} & 508 & 492 & 484 & 516 & 295 & 387 & 318 & 1000 & \multicolumn{1}{r}{} \bigstrut\\
\cline{1-10}    \end{tabular}%
  \label{tab:randomOptionsTotal}%
\end{table}
\end{landscape}%

\begin{landscape}
\begin{table}[htbp]
\footnotesize
  \centering
  \caption{Comparison (Set 2): Total summary.\\\footnotesize
  This table contains the summary of the performances of several methods for pricing options in Set 2. The assets follow equation \eqref{assetSolutionShifted} where the parameters are randomly generated and uniformly distributed in the following ranges: $r\in[0;0.1]$, $\sigma_i\in[0.1;0.6]$, $T\in [0.1;1]$, $S_0^{(i)}=[70;130]$, $a_i\in[-1;1]$, $\frac{\Strike}{\Basket}\in[0.95;1.05]$, $\shift_0^{(i)}\ee^{rT}\in[-20;20]$, $\sign_i\in[-1;1]$, $\JIntensityQ_i\in[0;0.2]$, $\logJumpMeanQ_i\in[-0.3;0.3]$ and $\logJumpVolatilityQ_i\in[0;0.3]$. 
  The results are showed (per column) along three different dimensions: risk-free rate, time to maturity and strike price. Along the different rows, the results per  method are showed: in particular, $BPW$ stands for the method in \cite{BorovkovaPermanaWeide2007}, $mGA$ and $mGB$ are considered for $m=4$ and $4GAB$ is a combination of $4GA$ and $4GB$. $4GAB$ returns the solution of the method
  that matches correctly the moments if only one of $4GA$ and $4GB$ works properly. The comparison for $4GAB$ is carried considering the error of the method that matches the
  moment if only one between $4GA$ and $4GB$ finds a solution or the worst error if both find a solution.}

    \begin{tabular}{|r|r|r|r|r|r|r|r|r|r|r|}
\cline{3-9}    \multicolumn{1}{r}{} &  & \multicolumn{2}{c|}{$r$} & \multicolumn{2}{c|}{$T$} & \multicolumn{3}{c|}{$K$} & \multicolumn{1}{r}{} & \multicolumn{1}{r}{} \bigstrut\\
\cline{3-10}    \multicolumn{1}{r}{} &  & \multicolumn{1}{c|}{$r\leq 0.05$} & \multicolumn{1}{c|}{$r>0.05$} & \multicolumn{1}{c|}{$T\leq 0.5$} & \multicolumn{1}{c|}{$T>0.5$} & \multicolumn{1}{c|}{$\frac{K}{\Basket_0}\leq 0.98$} & \multicolumn{1}{c|}{$0.98 < \frac{K}{\Basket_0}\leq 1.02$} & \multicolumn{1}{c|}{$\frac{K}{\Basket_0} >1.02 $} & \multicolumn{1}{c|}{Total} & \multicolumn{1}{r}{} \bigstrut\\
    \hline
        \multicolumn{1}{|c|}{\multirow{4}[8]{*}{\COne}} & BPW & 17.95\% & 17.44\% & 22.72\% & 12.37\% & 17.16\% & 16.83\% & 19.45\% & 17.70\% & BPW \bigstrut\\
    \cline{2-11}    \multicolumn{1}{|c|}{} &  $4GA$ & 78.70\% & 75.66\% & 68.74\% & 86.19\% & 78.22\% & 78.22\% & 74.74\% & 77.20\% &  $4GA$ \bigstrut\\
    \cline{2-11}    \multicolumn{1}{|c|}{} &  $4GB$ & 76.13\% & 77.89\% & 71.46\% & 82.89\% & 75.58\% & 78.22\% & 76.79\% & 77.00\% &  $4GB$ \bigstrut\\
    \cline{2-11}    \multicolumn{1}{|c|}{} &   $4GAB$ & \textbf{83.04}\% & \textbf{84.79\%} & \textbf{79.02\%} & \textbf{89.07}\% &\textbf{ 84.16}\% & \textbf{84.65\%} & \textbf{82.59\%} & \textbf{84.00\%} &  $4GAB$ \bigstrut\\
        \hline
    \multicolumn{1}{|c|}{\multirow{4}[8]{*}{\CTwo}} & BPW & 19.72\% & 24.14\% & 7.96\% & 36.70\% & 29.04\% & 17.82\% & 20.14\% & 21.90\% & BPW \bigstrut\\
\cline{2-11}    \multicolumn{1}{|c|}{} &  $4GA$ & 9.27\% & 13.18\% & 16.31\% & 5.77\% & 11.22\% & 10.15\% & 12.63\% & 11.20\% &  $4GA$ \bigstrut\\
\cline{2-11}    \multicolumn{1}{|c|}{} &  $4GB$ & 12.03\% & 11.36\% & 14.37\% & 8.87\% & 13.20\% & 11.88\% & 9.90\% & 11.70\% &  $4GB$ \bigstrut\\
\cline{2-11}    \multicolumn{1}{|c|}{} &  $4GAB$ & \textbf{2.96\%} & \textbf{3.85\%} & \textbf{3.11\%} & \textbf{3.71\%} & \textbf{3.96\%} & \textbf{2.97\%} & \textbf{3.41\%} & \textbf{3.40\%} &  $4GAB$ \bigstrut\\
    \hline
\multicolumn{1}{|c|}{\multirow{4}[8]{*}{\CThree}}      & BPW & 0.3731 & 0.3331 & 0.2729 & 0.4234 & 0.3623 & 0.3537 & 0.3454 & 0.3539 & BPW \bigstrut\\
\cline{2-11}    \multicolumn{1}{|c|}{} &  $4GA$ & 0.17 & \textbf{0.1545} & \textbf{0.096} & 0.2114 & \textbf{0.1772} & 0.1482 & 0.1656 & 0.1625 &  $4GA$ \bigstrut\\
\cline{2-11}    \multicolumn{1}{|c|}{} &  $4GB$ & \textbf{0.1632} & 0.1585 & 0.0969 & \textbf{0.2083} & 0.1841 & \textbf{0.1423} & \textbf{0.1593} & \textbf{0.1609} &  $4GB$ \bigstrut\\
\cline{2-11}    \multicolumn{1}{|c|}{}     &  $4GAB$ & 0.1784 & 0.1630 & 0.1116 & 0.2170 & 0.1966 & 0.1516 & 0.1678 & 0.1710 &  $4GAB$ \bigstrut\\
    \hline
    \multicolumn{2}{|c|}{\# options} & 507 & 493 & 515 & 485 & 303 & 404 & 293 & 1000 & \multicolumn{1}{r}{} \bigstrut\\
\cline{1-10}    \end{tabular}%
  \label{tab:groupBOverall}%
\end{table}
\end{landscape}%

\begin{landscape}
\begin{table}[htbp]
\footnotesize
  \centering
 \caption{Comparison: Delta-hedging performances.\\ \footnotesize
 This table contains the summary of the Delta-hedging performances of three methods.  BPW stands for the method in \cite{BorovkovaPermanaWeide2007} and $4GA$ and $4GB$ are the methods summarized in Table \ref{momentMatching}. The measures of error considered are: \CFour-- the volatility of Delta, \CFive-- the $MSE$ on the hedging performance along the life time of the contract, \CSix$\,$ and \CSeven-- the numbers of sub-hedging and super-hedging respectively, finally \CEight, \CNine$\,$ and \CTen-- respectively the average error on sub-hedging portfolios, super-hedging portfolios and all portfolios.}
    \begin{tabular}{|r|r|r|r|r|r|r|r|r|r|}
\cline{3-9}    \multicolumn{1}{r}{} &  & \multicolumn{1}{c|}{Basket 1} & \multicolumn{1}{c|}{Basket 2} & \multicolumn{1}{c|}{Basket $3^*$} & \multicolumn{1}{c|}{Basket $4^*$} & \multicolumn{1}{c|}{Basket $5^*$} & \multicolumn{1}{c|}{Basket $6^*$} & \multicolumn{1}{c|}{\multirow{2}[4]{*}{Total}} & \multicolumn{1}{r}{} \bigstrut\\
\cline{3-8}    \multicolumn{1}{r}{} &  & \multicolumn{1}{c|}{GBM} & \multicolumn{1}{c|}{GBM} & \multicolumn{1}{c|}{Shifted Jump} & \multicolumn{1}{c|}{Shifted Jump } & \multicolumn{1}{c|}{Shifted GBM} & \multicolumn{1}{c|}{Shifted GBM} & \multicolumn{1}{c|}{} & \multicolumn{1}{r}{} \bigstrut\\
    \hline
    \multicolumn{1}{|c|}{\multirow{7}[14]{*}{BPW}} & \multicolumn{1}{c|}{\CFour} & \textbf{0.1959} & 0.4707 & 0.2079 & 0.2332 & 0.2418 & \textbf{0.2045} & 0.259 & \multicolumn{1}{c|}{\CFour} \bigstrut\\
\cline{2-10}    \multicolumn{1}{|c|}{} & \multicolumn{1}{c|}{\CFive} & 1.5117 & 1.5622 & \textbf{1.5836} & \textbf{0.6354} & 2.1108 & 1.5385 & \textbf{1.4904} & \multicolumn{1}{c|}{\CFive} \bigstrut\\
\cline{2-10}    \multicolumn{1}{|c|}{} & \multicolumn{1}{c|}{\CSix} & \textbf{0.6457} & \textbf{0.1549} & 0.7042 & \textbf{0.2145} & 0.707 & \textbf{0.6641} & \textbf{51.51\%} & \multicolumn{1}{c|}{\CSix} \bigstrut\\
\cline{2-10}    \multicolumn{1}{|c|}{} & \multicolumn{1}{c|}{\CSeven} & 0.3543 & 0.8451 & \textbf{0.2958} & 0.7855 & \textbf{0.293} & 0.3359 & 48.49\%& \multicolumn{1}{c|}{\CSeven} \bigstrut\\
\cline{2-10}    \multicolumn{1}{|c|}{} & \multicolumn{1}{c|}{\CEight} & -3.6303 & -8.6137 & -4.0088 & -11.0397 & -6.3002 & \textbf{-3.2808} & -6.1456 & \multicolumn{1}{c|}{\CEight} \bigstrut\\
\cline{2-10}    \multicolumn{1}{|c|}{} & \multicolumn{1}{c|}{\CNine} & \textbf{1.676} & 14.5837 & \textbf{2.1672} & 15.8623 & \textbf{3.0774} & 2.2358 & 6.6004 & \multicolumn{1}{c|}{\CNine} \bigstrut\\
\cline{2-10}    \multicolumn{1}{|c|}{} & \multicolumn{1}{c|}{\CTen} & -1.7504 & 10.9898 & -2.1821 & 10.0914 & -3.5523 & \textbf{-1.428} & 2.0281 & \multicolumn{1}{c|}{\CTen} \bigstrut\\
    \hline
    \multicolumn{1}{|c|}{\multirow{7}[14]{*}{ $4GA$}} & \multicolumn{1}{c|}{\CFour} & 0.1984 & \textbf{0.2069} & \textbf{0.1986} & \textbf{0.1884} & \textbf{0.2395} & 0.2389 & \textbf{0.2118} & \multicolumn{1}{c|}{\CFour} \bigstrut\\
\cline{2-10}    \multicolumn{1}{|c|}{} & \multicolumn{1}{c|}{\CFive} & 1.502 & 1.335 & 1.6066 & 0.9208 & \textbf{2.0806} & \textbf{1.5351} & 1.4967 & \multicolumn{1}{c|}{\CFive} \bigstrut\\
\cline{2-10}    \multicolumn{1}{|c|}{} & \multicolumn{1}{c|}{\CSix} & 0.6511 & 0.6652 & \textbf{0.688} & 0.364 & \textbf{0.703} & 0.6934 & 62.74\% & \multicolumn{1}{c|}{\CSix} \bigstrut\\
\cline{2-10}    \multicolumn{1}{|c|}{} & \multicolumn{1}{c|}{\CSeven} & \textbf{0.3489} & \textbf{0.3348} & 0.312 & \textbf{0.636} & 0.297 & \textbf{0.3066} & \textbf{37.25\%} & \multicolumn{1}{c|}{\CSeven} \bigstrut\\
\cline{2-10}    \multicolumn{1}{|c|}{} & \multicolumn{1}{c|}{\CEight} & -3.5411 & -3.9376 & \textbf{-3.8939} & \textbf{-5.6858} & \textbf{-5.3825} & -3.7258 & \textbf{-4.3611} & \multicolumn{1}{c|}{\CEight} \bigstrut\\
\cline{2-10}    \multicolumn{1}{|c|}{} & \multicolumn{1}{c|}{\CNine} & 1.6796 & \textbf{1.7109} & 2.3771 & \textbf{1.4358} & 3.1234 & \textbf{2.1072} & \textbf{2.0723} & \multicolumn{1}{c|}{\CNine} \bigstrut\\
\cline{2-10}    \multicolumn{1}{|c|}{} & \multicolumn{1}{c|}{\CTen} & -1.7198 & -2.0466 & -1.9372 & -1.1567 & \textbf{-2.8564} & -1.9373 & \textbf{-1.9423} & \multicolumn{1}{c|}{\CTen} \bigstrut\\
    \hline
    \multicolumn{1}{|c|}{\multirow{7}[14]{*}{ $4GB$}} & \multicolumn{1}{c|}{\CFour} & 0.1983 & 0.207 & \textbf{0.1986} & 0.1886 & 0.2429 & 0.2389 & 0.2124 & \multicolumn{1}{c|}{\CFour} \bigstrut\\
\cline{2-10}    \multicolumn{1}{|c|}{} & \multicolumn{1}{c|}{\CFive} & \textbf{1.5007} & \textbf{1.3327} & 1.6066 & 0.9182 & 2.0832 & \textbf{1.5351} & 1.4961 & \multicolumn{1}{c|}{\CFive} \bigstrut\\
\cline{2-10}    \multicolumn{1}{|c|}{} & \multicolumn{1}{c|}{\CSix} & 0.6511 & 0.662 & \textbf{0.688} & 0.3608 & 0.7057 & 0.6934 & 62.68\% & \multicolumn{1}{c|}{\CSix} \bigstrut\\
\cline{2-10}    \multicolumn{1}{|c|}{} & \multicolumn{1}{c|}{\CSeven} & \textbf{0.3489} & 0.338 & 0.312 & 0.6392 & 0.2943 & \textbf{0.3066} & 37.32\% & \multicolumn{1}{c|}{\CSeven} \bigstrut\\
\cline{2-10}    \multicolumn{1}{|c|}{} & \multicolumn{1}{c|}{\CEight} & \textbf{-3.5273} & \textbf{-3.93} & \textbf{-3.8939} & -5.6936 & -5.5294 & -3.7258 & -4.3833 & \multicolumn{1}{c|}{\CEight} \bigstrut\\
\cline{2-10}    \multicolumn{1}{|c|}{} & \multicolumn{1}{c|}{\CNine} & 1.6796 & 1.711 & 2.3771 & 1.4564 & 3.1163 & \textbf{2.1072} & 2.0746 & \multicolumn{1}{c|}{\CNine} \bigstrut\\
\cline{2-10}    \multicolumn{1}{|c|}{} & \multicolumn{1}{c|}{\CTen} & \textbf{-1.7108} & \textbf{-2.0232} & \textbf{-1.9372} & \textbf{-1.1232} & -2.9845 & -1.9373 & -1.9527 & \multicolumn{1}{c|}{\CTen} \bigstrut\\
    \hline
    \end{tabular}%
  \label{tab:hedgingResults}%
\end{table}
\end{landscape} 

\begin{appendix}
\section*{Appendices} 
\def\appendixname{}

\section{Propositions Proofs}
\subsection{Proof of Proposition \ref{propositionShift}}
\label{AppendixShif}
\label{proof31}
Define the quantity
\begin{equation}
\Gamma(t)=e^{\left(r-\EJumpQ_i \JIntensityQ_i-\frac{1}{2}\sum_{j=1}^{\numWienerProcess}\coefficientWienerP{2}\right)t+\sum_{j=1}^{\numWienerProcess}\coefficientWienerP{} W_t^{(j)}}  \prod_{l=1}^{\PoissonP_t^{(i)}}{(\JumpP_l^{(i)}+1)}
\end{equation}
and calculate its expectation under the $\ProbQ$-martingale measure.
From equation \eqref{solutionArbitrageFree}, given that the system of equations \eqref{systemArbitrageFree} admits a solution, it follows that $\espQ[\Gamma(t)]=\ee^{rt}$.
Separating the right side of identity~\eqref{assetSolutionShiftedFirst} into two different components
\begin{equation}
\label{stockDivided}
S_t^{(i)}= [S_0^{(i)} \Gamma(t)]+\left[-\sign_i \Gamma(t) \shift_0^{(i)} + \sign_i\shift_t^{(i)}\right]
\end{equation}
and taking the discounted expectation of the quantity in the second parentheses lead to
\begin{eqnarray}
\label{proof2}
\ee^{-rt} \espQ\left[-\sign_i \Gamma(t) \shift_0^{(i)} + \sign_i\shift_t^{(i)}\right]&=&\ee^{-rt} \sign_i \left(-\espQ[\Gamma(t)]\shift_0^{(i)}+\espQ[\shift_t^{(i)}]\right) \nonumber\\
&=& \ee^{-rt}\sign_i \left(-\ee^{rt}\shift_0^{(i)}+\ee^{rt}\shift_0^{(i)}\right)=0
\end{eqnarray}
where we have used the martingale property of $\stocProcess{\shift_t^{(i)}}$. Consequently, the second bracket in~\eqref{stockDivided} does not influence the expectation but only the first  plays a role. Finally, by using \eqref{stockDivided} and \eqref{proof2},
\begin{equation*}
\espQ\left[\ee^{-rt}S_t^{(i)}\right]=\espQ\left[S_0^{(i)}\Gamma(t)\ee^{-rt}\right]=S_0^{(i)}
\end{equation*}
that concludes the proof.

Proposition \ref{propositionShift} can be generalized as follows.
\begin{proposition}\label{generalprop}
Proposition~\ref{propositionShift} still holds for any adapted process $\stocProcess{\shift_t^{(i)}}$ such that $\espQ[\ee^{-rt} \shift_t^{(i)}]=\shift_0^{(i)}$. In that case, the solution of the SDEs~\eqref{assedDinamicsShifted} is:
\begin{equation}
S_t^{(i)}=\left(S_0^{(i)}-b_i \shift_0^{(i)}\right)e^{\left(r-\beta_i \lambda_i-\frac{1}{2}\sum_{j=1}^{\numWienerProcess}\gamma_{i,j}^2\right)t+\sum_{j=1}^{\numWienerProcess}\gamma_{i,j} W_t^{(j)}}  \prod_{l=1}^{N_t^{(i)}}{(\JumpP_l^{(i)}+1)}+b_i\shift_t^{(i)}.
\end{equation}
\end{proposition}
\begin{proof}
The proof is identical to the one for  Proposition \ref{propositionShift} because we used there only the martingale property of the shift.
\end{proof}
\subsection{Proof of Proposition \ref{prop:moments}}
\label{proof41}
Formula \eqref{moments} is derived by exponentiation of formula~\eqref{shiftedBasket} where the moment generation function of $\sigma_i V_t^{(i)}  +\sum_{l=1}^{\PoissonP _t^{(i)}}{\log{(\JumpP_l^{(i)}+1)}}$ in \eqref{MultivariateMerton spec mgf} is calculated by conditioning with respect to $N_t$.

\subsection{Proof of Proposition \ref{propositionPrice}}
\label{proof42}
The proposition can be proved by considering the second equality in \eqref{pricingFormula}:
\begin{equation}\label{pricingFormulaAppendix}
c_0=\ee^{-rT}\espQ[(\ShiftedBasket_T-\ShiftedStrike)^+]\approx \ee^{-rT}\int_{\lowerLimit}^{\upperLimit}{\left[\ShiftedBasket_0\ee^{rT}(J(z)+\parameterA)-\ShiftedStrike\right]\phi(z)dz} 
\end{equation}
where, for $\ShiftedBasket_0>0$, $l_1=\zetaTilde$ and $l_2=+\infty$ and, for $\ShiftedBasket_0<0$, $l_1=-\infty$ and $l_2=\zetaTilde$. For the last integral in \eqref{pricingFormulaAppendix}, Appendix \ref{AppendixPriceFunction}  is useful.

\subsection{Proof of Proposition \ref{proposition43}}
\label{proof43}
The calculation of the hedging parameter can be achieved by direct differentiation of the approximate pricing formula \eqref{pricingFormulaAppendix} by applying Leibniz' rule. The results in Appendix \ref{AppendixA} are useful here.
\section{Computational Tools} 
\label{AppendixA}

\subsection{Tools for the pricing formula (Proposition \ref{propositionPrice})}
Hermite polynomials satisfy the recursive relation
\[
H_k(z)=z H_{k-1}(z) - H_{k-1}'(z) \quad k=1,2,\ldots 
\]
with $H_0(z)=1$.
 
Hence, for $\ShiftedBasket_0>0$
\[
\int_{\zetaTilde}^{+\infty} H_0(z)\phi(z) dz= \Phi(-\zetaTilde)
\]
and for $k\geq 1$
\[
\int_{\zetaTilde}^{+\infty} H_k(z)\phi(z) dz=\int_{\zetaTilde}^{+\infty}z H_{k-1}(z)\phi(z) dz - \int_{\zetaTilde}^{+\infty}H_{k-1}'(z)\phi(z) dz.
\]
Solving the second integral by parts and using $\phi'(z)=-z\phi(z)$,
\begin{eqnarray}
\int_{\zetaTilde}^{+\infty} H_k(z)\phi(z) dz&=&
\int_{\zetaTilde}^{+\infty}z H_{k-1}(z)\phi(z) dz +\nonumber\\ &&- \left[\left.-H_{k-1}(z)\phi(z)\right|_{\zetaTilde}^{+\infty}+\int_{\zetaTilde}^{+\infty}z H_{k-1}(z)\phi(z) dz\right]=\nonumber \\
&=& H_{k-1}({\zetaTilde})\phi({\zetaTilde}) \nonumber
\end{eqnarray}
and consequently,
\begin{eqnarray}\label{int1}
\int_{\zetaTilde}^{+\infty}{J(z)\phi(z)dz}&=& \functionF(\zetaTilde)+\coeffHermite{0} \Phi(-\zetaTilde).
\end{eqnarray}
Given the orthogonality feature of the Hermite polynomials,
\begin{equation}
\label{orthogonality}
\int_{-\infty}^\infty H_k(z)\phi(z)dx\,=\,0{\rm\ for\ }n\geq 1,
\end{equation}
for $\ShiftedBasket_0<0$
\begin{eqnarray}
\int_{-\infty}^{\zetaTilde} H_k(z)\phi(z) dz&=&-H_{k-1}({\zetaTilde})\phi({\zetaTilde}) \nonumber
\end{eqnarray}

and
\begin{eqnarray}\label{int2}
\int_{-\infty}^{\zetaTilde}{J(z)\phi(z)dz}&=& -\functionF(\zetaTilde)+\coeffHermite{0} \Phi(\zetaTilde).
\end{eqnarray}
\label{AppendixPriceFunction}

\subsection{Moments of the considered variables}
\label{momentY}
\label{momentNormalizedBasket}
The $k$-th moment of
\begin{equation*}
J(Z)=\sum_{k=0}^{m-1}{\coeffHermite{k} H_k(Z)}
\end{equation*}
can be calculated as:
\begin{equation}
\label{momentHermiteApprox}
\espQ[J^k]=\sum_{i_1=0}^{m}\ldots\sum_{i_k=0}^{m}{\coeffHermite{i_1}\ldots\coeffHermite{i_k}\esp[H_{i_1}(Z)\ldots H_{i_k}(Z)]}.
\end{equation}
Applying the property that the Hermite polynomials are orthogonal with respect to the standard normal probability density function (see equation \eqref{orthogonality}), formula~\eqref{momentHermiteApprox}  becomes $\espQ[J]=\coeffHermite{0}$, $\espQ[J^2]=\sum_{i=0}^{m}{i!\coeffHermite{i}^2 }$ and $\espQ[J^3]=\coeffHermite{0}^3 + (3 \coeffHermite{1}^2 + 6 \coeffHermite{2}^2 + 18 \coeffHermite{3}^2 + 72 \coeffHermite{4}^2 + 360 \coeffHermite{5}^2) \coeffHermite{0}+ 6 \coeffHermite{1}^2 \coeffHermite{2} + 36 \coeffHermite{1} \coeffHermite{2} \coeffHermite{3} + 144 \coeffHermite{1} \coeffHermite{3} \coeffHermite{4} + 720 \coeffHermite{1} \coeffHermite{4} \coeffHermite{5} + 8 \coeffHermite{2}^3 + 72 \coeffHermite{2}^2 \coeffHermite{4} + 108 \coeffHermite{2} \coeffHermite{3}^2 + 720 \coeffHermite{2} \coeffHermite{3} \coeffHermite{5} + 576 \coeffHermite{2} \coeffHermite{4}^2 + 3600 \coeffHermite{2} \coeffHermite{5}^2 + 648 \coeffHermite{3}^2 \coeffHermite{4}+ 8640 \coeffHermite{3} \coeffHermite{4} \coeffHermite{5}+ 1728 \coeffHermite{4}^3 + 43200 \coeffHermite{4} \coeffHermite{5}^2$.\\
For $k>3$, formula \eqref{momentHermiteApprox} can be evaluated in a closed form as a weighted sum of the moments of the standard normal variable knowing that the product between two Hermite polynomials is still a (non-Hermitian) polynomial and  that the expected value is a linear operator. The formulae are very long and are not given here for lack of space but they can be obtained upon request from the authors.

The $k$-th moment of the normalized basket for $mGA$ in Table~\ref{momentMatching} is given by:
\begin{equation}
\espQ\left[\left(\frac{B_T}{ \ShiftedBasket_0\ee^{rT}}\right)^k\right]=\frac{\esp[\ShiftedBasket^{k}_T]}{\ShiftedBasket_0^k \ee^{rkT}}
\end{equation}
and therefore the $k$-th moment of the normalized basket for $mGB$ is given by:
\begin{equation}
\espQ\left[\left(\frac{B_T}{ \ShiftedBasket_0\ee^{rT}}-1\right)^k\right]=\sum_{i=0}^{k}{\binom{k}{i}\frac{(-1)^i}{(\ShiftedBasket_0\ee^{rT})^{k-i}}\espQ[\ShiftedBasket^{k-i}_T]}
\end{equation}

\subsection{Hedging Parameters Calculations}
\label{appendix4}
For the hedging formulae, equations \eqref{int1} and \eqref{int2} are useful. 
These formulae can also be applied for the calculation of
\begin{equation*}
\int_{\lowerLimit}^{\upperLimit}{\frac{\partial J(z)}{\partial u}\phi(z)dz}
\end{equation*}
because
\begin{equation*}
\frac{\partial J(z)}{\partial u}=\sum_{k=0}^m{\frac{\partial \coeffHermite{k}}{\partial u}H_k(z)}
\end{equation*}
and, consequently, $J(Z)$ and $\frac{\partial J(z)}{\partial u}$ have the same structure but $\frac{\partial \coeffHermite{k}}{\partial u}$ takes the place of $\coeffHermite{k}$.

Finally, the derivatives $\frac{\partial \coeffHermite{k}}{\partial x}$ are calculated as below\footnote{This method is also used in \cite{BorovkovaPermanaWeide2007}.}. We start from the system:
\begin{numcases}{}
\espQ[J]=\espQ[X_T] \nonumber\\
\espQ[J^2]=\espQ[X_T^2]\nonumber\\
\quad \quad\cdots \nonumber\\
\espQ[J^m]=\espQ[X_T^m] \nonumber
\end{numcases}
where $X_T=\frac{\ShiftedBasket_T}{\ShiftedBasket_0 \ee^{rT}}+\parameterA$ and we differentiate left and right side of each equation with respect to $u$.

We are interested in the solution of the system when the derivatives are calculated in correspondence of the current status i.e. when the $\coeffHermite{i}$s are $\bar{\coeffHermite{0}},\bar{\coeffHermite{1}},\cdots,\bar{\coeffHermite{m}}$ and the parameter $u$ is $\bar{u}$. So we solve:
\begin{numcases}{}
\left.\frac{\partial\espQ[J]}{\partial u}\right|_{\bar{\coeffHermite{0}},\bar{\coeffHermite{1}},\cdots,\bar{\coeffHermite{m}}}=\left.\frac{\partial\espQ[X_T]}{\partial u}\right|_{\bar{u}} \nonumber\\
\left.\frac{\partial\espQ[J^2]}{\partial u}\right|_{\bar{\coeffHermite{0}},\bar{\coeffHermite{1}},\cdots,\bar{\coeffHermite{m}}}=\left.\frac{\partial\espQ[X_T^2]}{\partial u}\right|_{\bar{u}} \nonumber\\
\label{equ:system}\\
\quad \quad \quad \quad \cdots \quad \quad \quad \quad \quad \quad  \cdots\nonumber\\
\nonumber \\
\left.\frac{\partial\espQ[J^m]}{\partial u}\right|_{\bar{\coeffHermite{0}},\bar{\coeffHermite{1}},\cdots,\bar{\coeffHermite{m}}}=\left.\frac{\partial\espQ[X_T^m]}{\partial u}\right|_{\bar{u}} \nonumber
\end{numcases}
a linear system in the first derivative of the $\coeffHermite{i}$ with respect of $u$ calculate in correspondence of $\bar{\coeffHermite{0}},\bar{\coeffHermite{1}},\cdots,\bar{\coeffHermite{m}}$.
As before, the integrals can be evaluated in a closed-form. The quantities on the right of the equations \eqref{equ:system} are calculated differentiating the formula of the moments.
In particular, for the $\Delta$, the following relations are relevant:
\begin{equation*}
\frac{\partial \esp[X_T^k]}{\partial B_0}=\frac{\partial \esp[X_T^k]}{\partial a_1}\frac{\partial a_1}{\partial B_0}=\frac{\partial \esp[X_T^k]}{\partial a_1}\frac{1}{S_1}
\end{equation*}
and
\small
\begin{eqnarray}
\frac{\partial \esp[B_T^k]}{\partial a_1}&=& k a_1  \sum_{i_1=1}^N\cdots\sum_{i_{k-1}=1}^N
\left( a_{ i_1 } (S_0^{(i_1)}-\sign_{i_1}\shift_0^{(i_1)})  \ee^{  (r+\omega_{i_1})t    }\right)\times \cdots\nonumber \\
&\cdots&\times \left( a_{ i_{k-1} } (S_0^{(i_{k-1})}-\sign_{i_{k-1}}\shift_0^{(i_{k-1})})  \ee^{  (r+\omega_{i_{k-1}})t   }\right)
\mgf(\bm{e}_1+\bm{e}_{ i_1 } + \ldots + \bm{e}_{ i_{k-1} })\nonumber
\end{eqnarray} 
\end{appendix}

\end{document}